\documentclass[11pt]{article}
\usepackage[left=2.5cm,top=2.2cm, bottom=1.6in, right=2.5cm]{geometry}
\date{}
\usepackage{amsfonts}
\usepackage{amsmath, amsthm, amssymb}
\usepackage{graphicx}
\usepackage{hyperref, comment, cite}
\usepackage{psfrag}
\usepackage{bbm}
\usepackage{appendix}
\usepackage{mathrsfs}  

\usepackage{amssymb,amsmath,amsthm,latexsym,amscd,amsfonts}
\usepackage{authblk}
\usepackage{graphics}
\usepackage{cite}
\usepackage{epsfig}
\usepackage{epstopdf}
\usepackage{bbm}
\usepackage{dsfont}
\usepackage{setspace}
\usepackage{float}
\usepackage{relsize}
\usepackage{tikz}
\usetikzlibrary{%
  arrows,%
  shapes.misc,% wg. rounded rectangle
  shapes.arrows,%
  chains,%
  matrix,%
  positioning,% wg. " of "
  scopes,%
  decorations.pathmorphing,% /pgf/decoration/random steps | erste Graphik
  shapes.geometric,
  shadows,
  decorations.pathreplacing,
  patterns,
  shapes,
  decorations.markings
}

\newtheorem{mydef}{Definition}
\newtheorem{theorem}{Theorem}

\newtheorem{lemma}[theorem]{Lemma}
\newtheorem{example}[theorem]{Example}
\newtheorem{corollary}[theorem]{Corollary}
\newtheorem{proposition}[theorem]{Proposition}
\newtheorem{remark}[theorem]{Remark}

\def\VR{\kern-\arraycolsep\strut\vrule &\kern-\arraycolsep}
\def\vr{\kern-\arraycolsep & \kern-\arraycolsep}

\newcommand{\cb}{\text{\rm{cb}}}

\newcommand{\ket}[1]{|#1\rangle}
\newcommand{\bra}[1]{\langle#1|}

\newcommand{\bbE}{\mathbb{E}}    
\newcommand{\tr}{\text{\rm tr}}

\newcommand{\E}{\mathbb{E}}

\newcommand{\cN}{\mathcal{N}}
\newcommand{\cH}{\mathcal{H}}

\newcommand{\fL}{\mathbf L}

\definecolor{darkred}{rgb}{0.5,0,0}

\begin{document}

\title{A Correlation Measure Based on Vector-Valued $L_p$-Norms}
\author{\small Mohammad Mahdi Mojahedian, Salman Beigi, Amin Gohari, Mohammad Hossein Yassaee, Mohammad Reza Aref\thanks{This work was partially supported by Iran National Science Foundation (INSF) under contract No. 96/53979.}}

\maketitle

\vspace{-5mm}
\begin{abstract}
In this paper, we introduce a new measure of correlation for bipartite quantum states. This measure depends on a parameter $\alpha$, and is defined in terms of vector-valued $L_p$-norms. The measure is within a constant of the exponential of $\alpha$-R\'enyi mutual information, and reduces to the trace norm (total variation distance)  for $\alpha=1$. We will prove some decoupling type theorems in terms of this measure of correlation, and present some applications in privacy amplification as well as in bounding the random coding exponents. In particular, we establish a bound on the secrecy exponent of the wiretap channel (under the total variation metric) in terms of the $\alpha$-R{\'e}nyi mutual information according to \emph{Csisz\'ar's proposal}. 
%This  result complements the recent result of Yagli and Cuff for secrecy exponent of the i.i.d.\  codebook ensemble. 

\end{abstract}

%\begin{IEEEkeywords}
%Correlation metric.
%\end{IEEEkeywords}

% \fi
%\IEEEpeerreviewmaketitle
\tableofcontents

\section{Introduction}
\label{sec:intro}

%A function $V$ on bipartite quantum states 
 %$\rho_{AB}$ is called a \emph{measure of correlation} if it satisfies:
%\begin{enumerate}
%\item[(i)] \emph{Faithfulness:} $V(A;B)\geq 0$, and equality holds if and only if $\rho_{AB}=\rho_{A}\otimes\rho_{B}$.
%\item[(ii)] \emph{Monotonicity:} Local operations do not increase $V$. In other words, $V(X; Y)\leq V(A;B)$ if $\rho_{XY} = \Phi\otimes\Psi(\rho_{AB})$ where $\Phi_{A\rightarrow X}$ and $\Psi_{B\rightarrow Y}$ are completely positive trace-preserving (CPTP) maps. 
%\end{enumerate}
%Mutual information, R\'enyi mutual information and the trace distance between $\rho_{AB}$ and $\rho_{A}\otimes \rho_B$ are well-known examples of correlation measures. 

In this paper, for any $\alpha\geq 1$ we introduce a new measure of correlation $V_\alpha(A;B)$ by
\begin{align}\label{eq:def-v-alpha-0}
V_\alpha(A; B) = \Big\|  \big(I_B\otimes {\rho_A}^{-(\alpha-1)/2\alpha}\big)\rho_{BA}\big(I_B\otimes {\rho_A}^{-(\alpha-1)/(2\alpha)}\big) - \rho_B\otimes \rho_A^{1/\alpha}  \Big\|_{(1, \alpha)},
\end{align}
where $\|\cdot\|_{(1, \alpha)}$ denotes a certain norm which for $\alpha=1$ reduces to the $1$-norm. Since R\'enyi mutual information (according to Sibson's proposal) can also be expressed in terms of the $(1, \alpha)$-norm our measure of correlation is also related R\'enyi mutual information.
 
The main motivation for introducing these measures of correlation, particularly for $1\leq \alpha\leq 2$, is their applications in decoupling theorems.
% and privacy amplification (see e.g., \cite{Renner, DBWR14} and references therein). 
The point is that the average of $V_\alpha(A_0; B)$, when $\rho_{A_0B}$ is the outcome of a certain random CPTP map $\Phi_{A\rightarrow A_0}$ applied on the bipartite quantum state $\rho_{AB}$, can be bounded by $cV_\alpha(A; B)$ where $c<1$ is a constant. Thus our measures of correlation can be used to prove decoupling type theorems in information theory. 

Decoupling theorems have already found several applications in information theory. Most achievability results in quantum information theory are based on the phenomenon of decoupling (see~\cite{DBWR14} and references therein). Also, in classical information theory the OSRB method of~\cite{yassaee2014achievability} provides a similar decoupling-type tool for proving achievability results. The advantage of our decoupling theorem based on the measure $V_\alpha$, comparing to previous ones, is that it works for all values of $\alpha\in (1, 2]$. Given the relation between $V_\alpha$ and R\'enyi mutual information mentioned above, the parameters appearing in our decoupling theorem would be related to $\alpha$-R\'enyi mutual information, which for $\alpha=1$ reduces to Shannon's mutual information. Therefore, we can use our decoupling theorems not only for proving achievability results but also for proving interesting bounds on the \emph{random coding exponents}. We demonstrate this application via the examples of entanglement generation via a noisy quantum communication channel, and secure communication over a (classical) wiretap channel. In particular, we show a bound on the secrecy exponent of random coding over a wiretap channel in terms of R\'enyi mutual information according to \emph{Csisz\'ar's proposal}. 

% this property of $V_\alpha$ can be used for privacy amplification. More precisely,  our result and the result of Hayashi are based on different definitions of conditional R\'enyi entropy (while Shannon's conditional entropy has a unique definition, there are different extensions to the conditional R\'enyi entropy). 

%The special case of $\alpha=2$ and for classical systems finds a number of interesting representations. In particular, in this work are a new definition of mutual $f$-information is introduced and our measure of correlation is expressed in terms of the mutual $f$-information. 

Another application of our new measures of correlation is in secrecy. 
To measure the security of a communication system, one has to quantify the amount of information leaked to an eavesdropper. While the common security metric for measuring the leakage is mutual information (see \emph{e.g.}, see \cite{liang2009information}) or the total variation distance \cite{yassaee2014achievability, yagli2018exact},  there have been few recent works that motivate and define other measures of correlation to quantify leakage~\cite{Issa, Kamath, Li, Cuff,Weinberger,bellare2012semantic,dodis2005entropic, yu2017r}. Herein, we suggest the use of our metric instead of mutual information because it is a stronger metric and has a better rate-security tradeoff curve. To explain the rate-security tradeoff, consider a secure transmission protocol over a communication channel, achieving a communication rate of $R$ with certified leakage of at most $L$ according to the mutual information metric. Now, if the transmitter obtains a classified message for which leakage $L$ is no longer acceptable, it can sacrifice communication rate for improved transmission security. We show that the rate-security tradeoff with the mutual information metric is far worse than that of our metric. We will discuss this fact in more details via the problem of \emph{privacy amplification}.

The definition of our measure of correlation $V_\alpha(A; B)$ is based on the theory of \emph{vector-valued $L_p$ spaces}. These spaces are generalizations of the $L_p$ spaces and are defined via the \emph{theory of complex interpolation}. Then the proofs of our main theorems are heavily based on the interpolation theory. In particular, we use the Riesz-Thorin interpolation theorem several times, in order to establish an inequality for all $\alpha\in [1, 2]$ by interpolating between $\alpha=1$ and $\alpha=2$.

In the following section, we review some notations and introduce vector-valued $L_p$ norms. Section~\ref{sec:correlation_metric} introduces our new measure of correlation and presents some of its properties. 
Section~\ref{sec:decoupling} contains the main technical results of this paper. 
Section~\ref{sec:app} and Section~\ref{sec:exponent} contain some applications of our results in privacy amplification as well as in bounding the random coding exponents.

%**********************************************************

\section{Vector-valued $L_p$ norms}\label{sec:notation}

For a finite set $\mathcal A$ let $\ell( A)$ to be the vector space of functions $f:\mathcal A\rightarrow \mathbb C$. For any $p> 0$  and $f\in \ell( A)$ we define
$$\|f\|_p := \Big( \sum_{a\in \mathcal A} |f(a)|^p\Big)^{\frac1p}.$$
This quantity for $p\geq 1$ satisfies the triangle inequality and turns $\ell( A)$ into a normed space. The dual of $p$-norm is the $p'$-norm where $p'$ is the \emph{H\"older conjugate} of $p$ given by
\begin{align}\label{eq:holder-conj}
 \frac 1p + \frac 1{p'}=1.
\end{align}
More generally, for any $p, q, r>0$ with $1/p=1/q+1/r$ and any $f, g\in \ell( A)$ we have
$$\|fg\|_p\leq \|f\|_q\cdot \|g\|_r,$$
where $(fg)(a) =f(a)g(a)$.

Suppose that $\mathcal B$ is another set and we equip the vector space $\ell(B)$ with the $q$-norm. The question is how we can naturally define a $(p, q)$-norm on the space $\ell({AB}):=\ell({ A\times  B})=\ell( A)\otimes \ell( B)$ that is \emph{compatible} with the norm of the individual spaces $\ell( A), \ell( B)$. By compatible we mean that if $h=f\otimes g$ with $f\in \ell( A) $ and $g\in \ell( B)$ (i.e., $h(a, b) = f(a)g(b)$) then 
\begin{align}\label{eq:fg-pq-norm}
\|f\otimes g\|_{(p, q)} = \|f\|_p\cdot \|g\|_{q}.
\end{align}
To this end, any vector $h\in \ell({AB})$ can be taught of as a collection of $|\mathcal A|$ vectors $h_a\in \ell( B)$ for any $a\in \mathcal A$, where $h_a(b) = h(a, b)$. Let us denote $t(a) =\|h_a\|_q$. Then we may define 
$$\|h\|_{(p, q)} := \|t\|_p =\Big(  \sum_a \|h_a\|_q^p  \Big)^{1/p}.$$
This definition of the $(p,q)$-norm satisfies~\eqref{eq:fg-pq-norm}. Moreover, when $p=q$, this $(p,p)$-norm coincides with the usual $p$-norm.  
Finally, it is not hard to verify that the $(p, q)$-norm, for $p, q\geq 1$, is indeed a norm and satisfies the triangle inequality.

The $p$-norm can also be defined in the non-commutative case. Suppose that $\cH_A$ is a Hilbert space of finite dimension $d_A=\dim \cH_A$. Let $\fL(A)=\fL(\cH_A)$ to be the space of linear operators $M:\cH_A\rightarrow \cH_A$ acting on $\cH_A$. Again we can define
$$\|M\|_p = \Big(\tr(|M|^p)\Big)^{\frac 1p},$$
where $|M| = \sqrt{M^{\dagger}M}$, and $M^\dagger$ is the adjoint of $M$.  
For $p\geq 1$ this equips $\fL(A)$ with a norm, called the \emph{Schatten norm}, that satisfies the triangle inequality.  H\"older's inequality is also satisfied for Schatten norms~\cite{Bhatia-M}: if $p, q, r>0$ with $1/p=1/q+1/r$, then for $M, N\in \fL(A)$ we have
\begin{align}\label{eq:Holder-inequality}
\|MN\|_p\leq \|M\|_q\cdot \|N\|_r.
\end{align}

Our notation in the non-commutative case can be made compatible with the commutative case. 
By abuse of notation, an element $f_A\in \ell(A)$ can be taught of as a diagonal matrix of the form 
$$f_A = \sum_a f(a) \ket a\bra a,$$
acting on the Hilbert space $\cH_A$ with the orthonormal basis $\{\ket a:~ a\in \mathcal A\}$. Therefore, $\ell(A)$ can be taught of as a subspace of $\fL(A)$. We also have
$$\|f_A\|_p = \Big(\tr(|f_A|^p)\Big)^{\frac 1p} = \Big(\sum_a |f(a)|^p\Big)^{\frac 1p}.$$

Now the question is how we can define the $(p,q)$-norm in the non-commutative case. Let us start with the easy case of  $M_{AB} \in \ell(A)\otimes \fL(B)$. 
Then, following the above notation, $M_{AB}$ can be written as 
$$M_{AB} = \sum_{a} \ket a\bra a\otimes M_a,$$
with $M_{a} \in \fL(B)$. Similar to the fully commutative case we can define
\begin{align}\label{eq:p-q-c-q}
\|M_{AB}\|_{(p, q)} = \Big(  \sum_a \|M_{a}\|^p_q \Big)^{1/p}.
\end{align}

Now let us turn to the fully non-commutative case. In this case, the definition of the $(p, q)$-norm is not easy and is derived from \emph{interpolation theory}~\cite{Pisier}. Here, we present an equivalent definition provided in~\cite{Junge96 } (see also~\cite{DJKR}). We also focus on the case of $p\leq q$ that we need in this paper. In this case, since $p\leq q$ there exists $r\in (0, +\infty]$ such that $\frac{1}{p} = \frac1q + \frac 1r$. Then for any $M_{AB}\in \fL(AB)$ we define
\begin{align}\label{eq:def-NCVVN}
\|M_{AB}\|_{(p, q)} = \inf_{\sigma_A, \tau_A} \Big\|    \Big(\sigma_A^{-\frac{1}{2r}}\otimes I_B\Big) M_{AB} \Big(\tau_A^{-\frac{1}{2r}}\otimes I_B\Big)     \Big\|_q,
\end{align}
where the infimum is taken over all \emph{density matrices}\footnote{A density matrix is a positive semidefinite operator with trace one.} $\sigma_A, \tau_A\in \fL(A)$ and $I_B\in \fL(B)$ is the identity operator. In the following, for simplicity we sometimes suppress the identity operators in expressions of the form $\big(\sigma_A^{-1/2r}\otimes I_B\big) M_{AB}$ and write  $\sigma_A^{-1/2r}M_{AB}$. Therefore, 
\begin{align*}
\|M_{AB}\|_{(p, q)} = \inf_{\sigma_A, \tau_A} \Big\|    \sigma_A^{-\frac{1}{2r}} M_{AB} \tau_A^{-\frac{1}{2r}}     \Big\|_q.
\end{align*}
When $q\geq p\geq 1$, the $(p, q)$-norm satisfies the triangle inequality and is a norm. 
Some remarks are in line. 

\begin{remark} \emph{As in the commutative case, the order of subsystems in the above definition is important, i.e., $\|M_{AB}\|_{(p, q)}$ and $\|M_{BA}\|_{(p, q)}$ are different. }
\end{remark}

\begin{remark} \label{rem:2}\emph{
From H\"older's inequality~\eqref{eq:Holder-inequality}, one can derive that if $M_{AB}= M_A\otimes M_B$, then
$$\|M_A\otimes M_B\|_{(p, q)} =\|M_A\|_p\|M_B\|_q.$$}
\end{remark}

\begin{remark}\emph {When $M_{AB}\in \ell(A)\otimes \fL(B)$, the above definition of $(p, q)$-norm coincides with that of~\eqref{eq:p-q-c-q}. This can be shown by trying to optimize the choices of $\tau_A, \sigma_A$ in~\eqref{eq:def-NCVVN}, which can be taken to be diagonal. }
\end{remark}

\begin{remark}\label{rem:4}
 \emph{When $p=q$ the $(p, p)$-norm coincides with the usual $p$-norm~\cite{Pisier, Junge96}:
$$\| M_{AB}\|_{(p, p)}= \|M_{AB}\|_{p}.$$ }
\end{remark}

\begin{remark}\label{rem:5}
 \emph{ When $M_{AB}\geq 0$ is positive semidefinite, in~\eqref{eq:def-NCVVN} we may assume that $\sigma_A=\tau_A$, see~\cite{DJKR}. That is, when $M_{AB}$ is positive semidefinite we have
$$\|M_{AB}\|_{p, q} = \inf_{ \sigma_A} \Big\|    \sigma_A^{-\frac{1}{2r}} M_{AB} \sigma_A^{-\frac{1}{2r}}     \Big\|_q=\inf_{\sigma_A} \Big\|    \Gamma_{\sigma_A}^{-\frac1r} (M_{AB}) \Big\|_q,$$
where 
\begin{align}\label{eq:def-Gamma}
\Gamma_\sigma (X)= \sigma^{\frac12}X\sigma^{\frac12}.
\end{align}
}
\end{remark}

\medskip

We will compare our measure of correlation with R\'enyi mutual information which interestingly can also be written in terms of $(1, p)$-norms. 
For $\alpha\geq 1$ the \emph{sandwiched $\alpha$-R\'enyi relative entropy}  is defined by\footnote{All the logarithms in this paper are in base two.}
$$D_\alpha(\rho\|\sigma) = \alpha'\log \big\|\Gamma_{\sigma}^{-1/\alpha'}(\rho)\big\|_\alpha,$$
where $\alpha' = \alpha/(\alpha-1)$ is the H\"older conjugate of $\alpha$ given by~\eqref{eq:holder-conj}. The \emph{$\alpha$-R\'enyi mutual information}
 (Sibson's proposal)  for $\alpha>1$  is given by\footnote{See \cite{verdu2015alpha} for different definitions and properties of  R\'enyi mutual information.}
 $$I_\alpha(A; B) = \inf_{\sigma_B} D_\alpha(\rho_{AB}\| \rho_A\otimes \sigma_B).$$
Using the definition of $D_\alpha(\rho_{AB}\| \rho_A\otimes \sigma_B)$ and Remark~\ref{rem:5} we find that
$$I_\alpha(A; B) =\alpha' \log \Big\|\Gamma_{\rho_A}^{-1/\alpha'}(\rho_{BA}) \Big \|_{(1, \alpha)}.$$
In particular, for classical random variables $A$ and $B$ with joint distribution $p_{AB}$ we have
\begin{align}\label{eq:RMI-sibson}
I_{\alpha}(A;B)&=\alpha' \log\bigg(\sum_b \Big[\sum_a p(a)p(b|a)^{\alpha}\Big]^{1/\alpha}\bigg).
\end{align}
Finally the
 $\alpha$-R\'enyi conditional entropy is defined by
\begin{align}H_\alpha(A|B) = - \inf_{\sigma_B} D_\alpha(\rho_{AB} \| I_A\otimes \sigma_B) = -\alpha' \log \|\rho_{BA}\|_{(1, \alpha)}.\label{eqn-cond-entr-ren}\end{align}

We finish this section by stating a lemma about the monotonicity of the $(1,\alpha)$-norm.

\begin{lemma} \label{lemma1}
For any $M_{AB}$ and any density matrix $\xi_A$ the function
$\alpha\mapsto \big\|\Gamma_{\xi_A}^{-1/\alpha'}(M_{BA})\big\|_{(1, \alpha)}$ is non-decreasing on $[1, +\infty)$.
\end{lemma}

\begin{proof}
Let $\beta> \alpha\geq 1$, and let $\gamma>0$ be such that $1/\alpha=1/\beta+1/\gamma$. 
Using H\"older's inequality for arbitrary density matrices $\sigma_B, \tau_B$ we have
\begin{align*}
\Big\|\sigma_B^{-1/(2\alpha')} \Gamma_{\xi_A}^{-1/\alpha'}&(M_{BA}) \tau_B^{-1/(2\alpha')}\Big\|_\alpha \\
& = \Big\| \big(\sigma_B\otimes \xi_A\big)^{1/(2\gamma)} \sigma_B^{-1/(2\beta')} \Gamma_{\xi_A}^{-1/\beta'} (M_{BA}) \tau_B^{-1/(2\beta')}\big( \tau_B\otimes \xi_A\big)^{1/(2\gamma)}\Big\|_\alpha\\
& \leq\Big\|\big(\sigma_B\otimes \xi_A\big)^{1/(2\gamma)}\Big\|_{2\gamma}\cdot \Big \|  \sigma_B^{-1/(2\beta')} \Gamma_{\xi_A}^{-1/\beta'} (M_{BA}) \tau_B^{-1/(2\beta')} \Big\|_\beta\cdot \Big\|\big( \tau_B\otimes \xi_A\big)^{1/(2\gamma)}\Big\|_{2\gamma}\\
& =  \Big \|  \sigma_B^{-1/(2\beta')} \Gamma_{\xi_A}^{-1/\beta'} (M_{BA}) \tau_B^{-1/(2\beta')} \Big\|_\beta.
\end{align*}
Taking infimum over $\sigma_B, \tau_B$ we obtain the desired result.
\end{proof}

\subsection{Completely bounded norm}\label{app:CB-norm}

The \emph{completely bounded norm} of a super-operator $\Phi: \fL(A)\rightarrow \fL(B)$ is defined by
$$\|\Phi\|_{\cb, p\rightarrow q} := \sup_{d_C} \big\|\mathcal I_C\otimes \Phi\big\|_{(\infty, p)\rightarrow (\infty, q)} = \sup_{X_{CA}} \frac{\big\|\mathcal I_C\otimes \Phi(X_{CA})\big\|_{(\infty, q)}}{\big\|X_{CA}\big\|_{(\infty, p)}},$$
where the supremum is taken over all auxiliary Hilbert spaces $\cH_C$ with arbitrary dimension $d_C$ and $\mathcal I_C:\fL(C)\to \fL(C)$ is the identity super-operator. 
In the above definition, we may replace $\infty$ with any $1\leq t\leq \infty$, see~\cite{Pisier}. That is, for any $t\geq 1$ we have
\begin{align}\label{eq:cb-norm-t}
\|\Phi\|_{\cb, p\rightarrow q} := \sup_{d_C} \big\|\mathcal I_C\otimes \Phi\big\|_{(t, p)\rightarrow (t, q)}.
\end{align}
We say that a super-operator between spaces with certain norms is a \emph{complete contraction} if its completely bounded norm is at most $1$. 

\begin{lemma}\label{eq:SWAP-norm}
For any $M_{BCA} \in \fL(BCA)$  and $1\leq \alpha\leq \infty$ we have
$$\|M_{BCA}\|_{(1, 1, \alpha)} \geq \|M_{BCA}\|_{(1, \alpha, \alpha)}.$$
\end{lemma} 

\begin{proof}
First of all the swap super-operator is a complete contraction~\cite{Pisier}, i.e.,  
$$\big\|M_{BCA}\big\|_{(1, 1, \alpha)}\geq \big\|M_{BAC}\big\|_{(1, \alpha, 1)}.$$
Therefore, it suffices to show that 
$$\big\|M_{BAC}\big\|_{(1, \alpha, 1)}\geq \big\|M_{BCA}\big\|_{(1, \alpha, \alpha)} =\big \|M_{BAC}\big\|_{(1, \alpha, \alpha)}.$$
Equivalently we need to show that 
$$\big\|\mathcal I_{AC}\big\|_{\cb, (\alpha, \alpha)\rightarrow (\alpha, 1)}  \leq 1.  $$
Using~\eqref{eq:cb-norm-t} we have
$$\big\|\mathcal I_{AC}\big\|_{\cb, (\alpha, \alpha)\rightarrow (\alpha, 1)} = \sup_{d_{E}} \big\|\mathcal I_{EAC} \big\|_{(\alpha, \alpha, \alpha)\rightarrow (\alpha, \alpha, 1)} = \sup_{d_{D}} \big\| \mathcal I_{DC} \big\|_{(\alpha, \alpha)\rightarrow (\alpha, 1)} = \big\| \mathcal I_C\big\|_{\cb, \alpha\rightarrow 1} .$$
Next since $\mathcal I_C$ is completely positive and $\alpha\geq 1$ we have~\cite{DJKR}
$$\big\|\mathcal I_C\big\|_{\cb, \alpha\rightarrow 1}  = \big\|\mathcal I_C\big\|_{\alpha\rightarrow 1} =1.$$
We are done.

\end{proof}

%%%%%%%%%%%%%%%%%%%%%%%%%%%%%%%%%%%%%%%%%%%%%%%%%%%%%%%%%%%%
%%%%%%%%%%%%%%%%%%%%%%%%%%%%%%%%%%%%%%%%%%%%%%%%%%%%%%%%%%%%

\section{A new measure of correlation}\label{sec:correlation_metric} 

In this section, we define our measure of correlation and study some of its properties. 

\begin{mydef}
Let $\rho_{AB}$ be an arbitrary bipartite density matrix. For any $\alpha\geq 1$ we define\footnote{When $\alpha=1$ we have $\alpha'=+\infty$.}
\begin{align}V_\alpha(A; B) &:= \Big\|  \Gamma_{\rho_A}^{-1/\alpha'}(\rho_{BA}) - \rho_B\otimes \rho_A^{1/\alpha}  \Big\|_{(1, \alpha)},\\
W_\alpha(A|B)&:=  \Big\|\rho_{BA} - \rho_B\otimes \frac{I_A}{d_A}\Big\|_{(1,\alpha)},\end{align}
where $1/\alpha+1/\alpha'=1$, and $\rho_A=\tr_B(\rho_{AB}), \rho_B=\tr_A(\rho_{AB})$ are the marginal states on $A$ and $B$ subsystems, respectively. 
\end{mydef}

As will be seen below, $V_\alpha(A;B)$ is a measure of correlation while $W_\alpha(A|B)$ is a related quantity that may be thought of as a conditional entropy. 

By Remark~\ref{rem:4} when $\alpha=1$, $V_\alpha$ and $W_\alpha$ can be expressed in terms of the $1$-norm:
\begin{align}
V_1(A; B) &= \|\rho_{AB} - \rho_A\otimes \rho_B\|_1,\\
W_1(A|B) &=\Big\|\rho_{BA} - \rho_B\otimes \frac{I_A}{d_A}\Big\|_{1}.
\end{align}

%Using Lemma \ref{lemma1}, 
%$\alpha\mapsto V_\alpha(A; B)$ and $\alpha\mapsto W_{\alpha}(V|E)$ are non-decreasing and hence for any $\alpha\geq 1$ 
%\begin{align}
%V_\alpha(A; B) &\geq \|\rho_{AB} - \rho_A\otimes \rho_B\|_1,\label{eq:V_TV_relation}\\
%W_\alpha(A|B) &\geq \Big\|\rho_{BA} - \rho_B\otimes \frac{I_A}{d_A}\Big\|_{1}.
%\end{align}

In the classical case when $p_{AB}$ is a joint probability distribution
we have
$$V_\alpha(A;B) = \sum_{b} \Big(   \sum_a p(a) \big|p(b|a) - p(b)     \big|^{\alpha}    \Big)^{1/\alpha},$$
and
$$W_\alpha(A|B) = \sum_{b} p(b) \bigg(   \sum_a \Big|p(a|b) - \frac{1}{|\mathcal A|}     \Big|^{\alpha}    \bigg)^{1/\alpha}.$$

As an immediate property of the above definitions, both 
$V_\alpha(A; B)$ and $W_\alpha(A|B)$ are non-negative. Moreover, since they are defined in terms of a norm, we have
$V_\alpha(A; B)=0$ if and only if $\rho_{AB}=\rho_A\otimes \rho_B$, and $W_\alpha(A|B) =0$ if and only if $\rho_{AB} = \frac{I_A}{d_A}\otimes \rho_B$.  %In particular we can think of $W_\alpha(A|B)$ as the distance from the ideal source of randomness.

\begin{proposition}\label{prop:non-decreasing}
For any $\rho_{AB}$
the functions 
$$\alpha\mapsto V_\alpha(A; B),$$
and
$$\alpha\mapsto d_A^{\frac {1}{\alpha'}}W_\alpha(A; B),$$
are non-decreasing. In particular, for any $\alpha\geq 1$ we have
$$V_\alpha(A; B)\geq \|\rho_{AB} - \rho_A\otimes \rho_B\|_1,  \qquad \text{ and } \qquad W_\alpha(A; B) \geq d_A^{-\frac{1}{\alpha'}} \Big\|\rho_{AB} - \frac{I_A}{d_A}\otimes \rho_B\Big\|_1. $$
\end{proposition}

\begin{proof}
For the monotonicity of $\alpha\mapsto V_\alpha(A; B)$, in Lemma~\ref{lemma1} put $M_{AB} = \rho_{AB} - \rho_A\otimes \rho_B$ and $\xi_A = \rho_A$. For the other monotonicity let $M_{AB}= \rho_{AB} - I_A/d_A\otimes \rho_B$ and $\xi_A= I_A/d_A$.
\end{proof}

We now prove the main property of $V_{\alpha}(A;B)$ and $W_\alpha(A|B)$, namely their monotonicity under local operations.

\begin{theorem}[Monotonicity under local operations] \label{thm:monotonicity}

\begin{itemize}
\item[{\rm (i)}] For any $\rho_{AB}$ and all CPTP maps $\Phi_{A\rightarrow X}$ and $\Psi_{B\rightarrow Y}$ we have 
$$V_\alpha(X; Y)\leq V_\alpha(A;B),$$ 
where $\rho_{XY} = \Phi\otimes\Psi(\rho_{AB})$.
\item[{\rm (ii)}]
For any $\rho_{AB}$ and any CPTP map $\Psi_{B\rightarrow Y}$
$$W_{\alpha}(A|Y)\leq W_\alpha(A|B),$$
where $\rho_{AY} = \mathcal I_A\otimes \Psi(\rho_{AB})$ and $\mathcal I_A$ is the identity super-operator. 
\end{itemize}
\end{theorem}

\begin{proof}
For (i) 
we compute
\begin{align*}
V_\alpha(X; Y) & = \Big\|    \Gamma_{\rho_{X}}^{-1/\alpha'}(\rho_{YX})  - \rho_Y\otimes \rho_{X}^{1/\alpha}  \Big\|_{(1, \alpha)}\\
& = \Big\|  \Big( \Psi\otimes \Gamma_{\Phi(\rho_A)}^{-1/\alpha'}\circ \Phi\circ \Gamma_{\rho_A}^{1/\alpha'}   \Big)  \big(  \Gamma_{\rho_A}^{-1/\alpha'} (\rho_{BA}) - \rho_B\otimes \rho_A^{1/\alpha}   \big)      \Big\|_{(1, \alpha)}\\
&\leq \Big\| \Psi\otimes \Gamma_{\Phi(\rho_A)}^{-1/\alpha'}\circ \Phi\circ \Gamma_{\rho_A}^{1/\alpha'} \Big  \|_{(1, \alpha)\rightarrow (1, \alpha)}\cdot \Big\|   \Gamma_{\rho_A}^{-1/\alpha'} (\rho_{BA}) - \rho_B\otimes \rho_A^{1/\alpha}        \Big\|_{(1, \alpha)}\\
&= \Big\| \Psi\otimes \Gamma_{\Phi(\rho_A)}^{-1/\alpha'}\circ \Phi\circ \Gamma_{\rho_A}^{1/\alpha'} \Big  \|_{(1, \alpha)\rightarrow (1, \alpha)}\cdot V_\alpha(A; B)\\
&= \Big\| \Psi\otimes \mathcal I_A\Big  \|_{(1, \alpha)\rightarrow (1, \alpha)}\cdot \Big\| \mathcal I_B\otimes \Gamma_{\Phi(\rho_A)}^{-1/\alpha'}\circ \Phi\circ \Gamma_{\rho_A}^{1/\alpha'} \Big\|_{(1, \alpha)\rightarrow (1, \alpha)}\cdot V_\alpha(A; B),
\end{align*}
where $(1, \alpha)\rightarrow (1, \alpha)$ denotes the super-operator norm:
$$\|\mathcal T\|_{(1, \alpha)\rightarrow (1, \alpha)} := \sup_{M\neq 0}   \frac{\|\mathcal T(M)\|_{(1, \alpha)}}{\| M\|_{(1, \alpha)}}.$$
Now using equation~(3.5) and Theorem 13 of~\cite{DJKR} we have
$$\Big\| \mathcal I_B\otimes \Gamma_{\Phi(\rho_A)}^{-1/\alpha'}\circ \Phi\circ \Gamma_{\rho_A}^{1/\alpha'} \Big  \|_{(1, \alpha)\rightarrow (1, \alpha)} \leq  \Big\| \Gamma_{\Phi(\rho_A)}^{-1/\alpha'}\circ \Phi\circ \Gamma_{\rho_A}^{1/\alpha'} \Big  \|_{\alpha\rightarrow \alpha}.$$
On the other hand, using Lemma~9 of~\cite{BD} (see also~\cite{Beigi}) we have  
$$ \Big\| \Gamma_{\Phi(\rho_A)}^{-1/\alpha'}\circ \Phi\circ \Gamma_{\rho_A}^{1/\alpha'} \Big  \|_{\alpha\rightarrow \alpha}\leq 1.$$
Moreover, by Lemma 5 of~\cite{DJKR} we have
$$ \|\Psi\otimes \mathcal I_A \|_{(1, \alpha)\rightarrow (1, \alpha)} =  \|\Psi \|_{1\rightarrow 1}=1,$$
since $\Psi$ is CPTP. We conclude that, $V_\alpha(X; Y)\leq V_\alpha(A;B)$.

The proof of (ii) is similar, so we skip it. 
\end{proof}

We now state the relation between $V_\alpha, W_\alpha$ and R\'enyi information measures. 

\begin{proposition}\label{prop2}
For any bipartite density matrix $\rho_{AB}$ we have
$$2^{\frac{1}{\alpha'} I_\alpha(A; B)}-1\leq V_\alpha(A; B) \leq 2^{\frac{1}{\alpha'} I_\alpha(A; B)} +1, $$
where $\alpha'$ is the H\"older conjugate of $\alpha$.
For $W_\alpha(A|B)$ we have
$$2^{-\frac{1}{\alpha'}H_\alpha(A|B)} - d_A^{-\frac{1}{\alpha'}} \leq W_\alpha(A|B) \leq 2^{-\frac{1}{\alpha'}H_\alpha(A|B)} + d_A^{-\frac{1}{\alpha'}},$$
where $d_A=\dim \cH_A$.
\end{proposition}

\begin{proof}
By the triangle inequality we have
\begin{align*}
\big\|  \Gamma_{\rho_A}^{-1/\alpha'}(\rho_{BA})\big\|_{(1, \alpha)} - \big\|\rho_B\otimes \rho_A^{1/\alpha}  \big\|_{(1, \alpha)}\leq \big\| & \Gamma_{\rho_A}^{-1/\alpha'}(\rho_{BA}) -  \rho_B\otimes \rho_A^{1/\alpha}  \big\|_{(1, \alpha)}\\
&\leq \big\|  \Gamma_{\rho_A}^{-1/\alpha'}(\rho_{BA})\big\|_{(1, \alpha)} +\big\| \rho_B\otimes \rho_A^{1/\alpha}  \big\|_{(1, \alpha)}.
\end{align*}
Moreover, by Remark~\ref{rem:2} we have
$$ \big\|\rho_B\otimes \rho_A^{1/\alpha}  \big\|_{(1, \alpha)} =  \big\|\rho_B\big\|_1\cdot\big\| \rho_A^{1/\alpha}  \big\|_{\alpha}=1.$$
These give the first inequality. 
The proof of the second inequality is similar.
\end{proof}

\begin{theorem}\label{thm:conditional-V-1}
Let $\rho_{ABC}$ be a tripartite density matrix. Then the followings hold:
\begin{itemize}
\item[{\rm (i)}] For any $1\leq \alpha\leq 2$ we have
$$W_\alpha(A|BC)\leq 2^{\frac2\alpha -1} d_C^{\frac{1}{\alpha'}} W_\alpha(AC|B)$$
\item[{\rm (ii)}] Assume that $\rho_{AC} = \frac{1}{d_Ad_C} I_A\otimes I_C$. Then for any $1\leq \alpha\leq 2$ we have
$$V_\alpha(A; BC) \leq 2^{\frac2\alpha -1} V_\alpha(AC; B).$$
Moreover, if $C$ is classical (and $\rho_{AC} = \frac{1}{d_Ad_C} I_A\otimes I_C$) then
$$V_\alpha(A; B|C) \leq 2^{\frac2\alpha -1} V_\alpha(AC; B),$$
where we define
$$V_\alpha(A; B|C) = \sum_c p(c) V_\alpha(A; B| C=c).$$

\end{itemize}
\end{theorem}
\begin{proof}
The proof of (ii) is immediate once we have (i) since if $\rho_{A} = I_A/d_A$ then
$$V_\alpha(A; B) = d_A^{1/\alpha'} W_\alpha(A|B).$$
Moreover, when $C$ is classical and
$$\rho_{ABC} = \sum_c p(c) \rho_{AB|c} \otimes \ket c\bra c,$$
with $\rho_{A|c} =\tr_{B} (\rho_{AB|c}) =\rho_A$,
we have
\begin{align*}
V_\alpha(A; B|C) & = \sum_c p(c) \Big\| \Gamma_{\rho_A}^{-1/\alpha'}\big(\rho_{AB|c}\big) - \rho_{B|c}\otimes \rho_A^{1/\alpha} \Big\|_{(1, \alpha)}\\
& = \Big\| \sum_c p(c) \ket c\bra c\otimes \Gamma_{\rho_A}^{-1/\alpha'}\big(\rho_{AB|c}\big) - \sum_c p(c)\ket c\bra c\otimes \rho_{B|c}\otimes \rho_A^{1/\alpha}   \Big\|_{(1, 1, \alpha)}\\
& = \Big\|  \Gamma_{\rho_A}^{-1/\alpha'}(\rho_{CBA}) - \rho_{CB}\otimes \rho_A^{1/\alpha}  \Big\|_{(1, 1, \alpha)}\\
& = V_\alpha(A; BC).
\end{align*}
So we only need to prove (i). 

Define $\Xi: \fL(BCA)\rightarrow \fL(BCA)$ by
$$\Xi(M_{BCA}) = M_{BCA} - \tr_A(M_{BCA})\otimes I_A/d_A.$$
We claim that 
\begin{align}\label{eq:norm-Xi-2}
\big\|\Xi\big\|_{(1, \alpha, \alpha)\rightarrow (1, 1, \alpha)}\leq 2^{\frac{2}{\alpha}-1} d_C^{\frac{1}{\alpha'}}.
\end{align}
Since vector valued $L_p$-spaces form an \emph{interpolation family}~\cite{Pisier}, by the Riesz-Thorin theorem (see Appendix~\ref{app:interpolation}) it suffices to prove this for $\alpha=1$ and $\alpha=2$.
For $\alpha=1$ by the triangle inequlality we have
\begin{align*}
\big\|M_{BCA} - \tr_A(M_{BCA})\otimes I_A/d_A\big\|_{1}
& \leq \big\|M_{BCA}\big\|_1 + \big\| \tr_A(M_{BCA})\otimes I_A/d_A   \big\|_1
\\
& = \big\|M_{BCA}\big\|_1 + \big\| \tr_A(M_{BCA})\big\|_1\cdot \big\| I_A/d_A\big\|_1\\
& = \big\|M_{BCA}\big\|_1 + \big\| \tr_A(M_{BCA})\big\|_1\\
& \leq 2 \big\|M_{BCA}\big\|_1,
\end{align*}
where the second inequality comes from the fact that $\|\tr_A\|_{1\rightarrow 1}\leq 1$ that is easy to verify.
%\begin{align*}
%\big\| M_{BCA})\otimes I_A/d_A   \big\|_1 &= \sup_{X_{BCA}:\, \|X\|_\infty\leq 1} \big|\tr\big(X_{BCA}\cdot  \tr_A(M_{BCA})\otimes I_A/d_A\big)  \big| \\
%& \leq \sup_{X_{BC}:\, \|X\|_\infty\leq 1} \big|\tr\big(X_{BC}\otimes I_A\cdot  \tr_A(M_{BCA})\otimes I_A/d_A\big)  \big| 
% \end{align*}
We now prove the inequality for $\alpha=2$. We compute
\begin{align*}
\big\|  M_{BCA} - &\tr_A(M_{BCA})\otimes I_A/d_A  \big\|_{(1, 1, 2)}^2\\
&=\inf_{\tau_{BC}, \sigma_{BC}}\Big\| \tau_{BC}^{-1/4} M_{BCA} \sigma_{BC}^{-1/4}  - \tau_{BC}^{-1/4} \tr_A\big( M_{BCA}\sigma_{BC} \big)\sigma_{BC}^{-1/4} \otimes I_{A}/d_A\Big\|_2^2 \\ 
&=\inf_{\tau_{BC}, \sigma_{BC}} \big\|  \tau_{BC}^{-1/4} M_{BCA} \sigma_{BC}^{-1/4} \big\|_2^2 - \frac{1}{d_A}
\big\|  \tau_{BC}^{-1/4} \tr_A\big( M_{BCA}\sigma_{BC} \big)\sigma_{BC}^{-1/4}\big\|_2^2\\
& \leq \inf_{\tau_{BC}, \sigma_{BC}}\big\|  \tau_{BC}^{-1/4} M_{BCA} \sigma_{BC}^{-1/4} \big\|_2^2\\
& \leq \inf_{\tau_{B}, \sigma_{B}}\big\|  (\tau_{B}\otimes I_C/d_C)^{-1/4} M_{BCA} (\sigma_{B}\otimes I_C/d_C)^{-1/4} \big\|_2^2\\
& = \inf_{\tau_{B}, \sigma_{B}}  d_C \big\|  \tau_{B}^{-1/4} M_{BCA} \sigma_{B}^{-1/4} \big\|_2^2\\
& = d_C\|M_{BCA}\|_{(1, 2, 2)}^2.
\end{align*}
Then~\eqref{eq:norm-Xi-2} holds for all $1\leq \alpha\leq 2$ and for any $M_{BCA}$ we have
$$\big\|  M_{BCA} - \tr_A(M_{BCA})\otimes I_A/d_A  \big\|_{(1, 1, \alpha)}\leq 2^{\frac{2}{\alpha}-1} d_C^{\frac{1}{\alpha'}}\|M_{BCA}\|_{(1, \alpha, \alpha)}.$$
Letting 
$$M_{BCA} = \rho_{BCA} - \rho_B\otimes \frac{I_C}{d_C}\otimes\frac{I_A}{d_A},$$
in the above inequality we obtain the desired result. 

\end{proof}

The next theorem gives a ``weak converse" of the above inequalities.

\begin{theorem}\label{thm:conditional-V-2}
For every tripartite density matrix $\rho_{ABC}$ and $\alpha\geq 1$ the followings hold:
\begin{itemize}
\item[{\rm (i)}] $W_\alpha(AC|B)\leq W_\alpha(A|BC) +d_A^{-1/\alpha'} W_\alpha(C|B)$. 
\item[{\rm (ii)}] If $\rho_{AC} = \frac{1}{d_Ad_C}I_A\otimes I_C$ then
$$V_\alpha(AC; B) \leq d_C^{-1/\alpha'}V_\alpha(A; BC) +  V_\alpha(C; B).$$
Moreover if $C$ is classical (and $\rho_{AC} = \frac{1}{d_Ad_C}I_A\otimes I_C$) then
$$V_\alpha(AC; B) \leq d^{-1/\alpha'}_CV_\alpha(A; B|C) + V_\alpha(C; B).$$
\end{itemize}
\end{theorem}

\begin{proof}
Again we only need to prove (i). To this end we use the triangle inequality as follows:
\begin{align*}
W_\alpha(AC|B) & = \Big\| \rho_{BCA} - \rho_B\otimes \frac{I_C}{d_C}\otimes \frac{I_A}{d_A}\Big\|_{(1, \alpha, \alpha)}\\
& \leq \Big\| \rho_{BCA} - \rho_{BC}\otimes \frac{I_A}{d_A}\Big\|_{(1, \alpha, \alpha)} + \Big\| \rho_{BC}\otimes \frac{I_A}{d_A} - \rho_B\otimes \frac{I_C}{d_C}\otimes \frac{I_A}{d_A}\Big\|_{(1, \alpha, \alpha)}\\
& = \Big\| \rho_{BCA} - \rho_{BC}\otimes \frac{I_A}{d_A}\Big\|_{(1, \alpha, \alpha)} + \big\|  \frac{I_A}{d_A}\big\|_\alpha\cdot\Big\| \rho_{BC} - \rho_B\otimes \frac{I_C}{d_C}\Big\|_{(1, \alpha)}\\
& \leq \Big\| \rho_{BCA} - \rho_{BC}\otimes \frac{I_A}{d_A}\Big\|_{(1, 1, \alpha)} + d_A^{-1/\alpha'}\cdot\Big\| \rho_{BC} - \rho_B\otimes \frac{I_C}{d_C}\Big\|_{(1, \alpha)}\\
& = W_\alpha(A|BC) + d_A^{-1/\alpha'} W_\alpha(C|B),
\end{align*}
where the last inequality follows from Lemma~\ref{eq:SWAP-norm}.
% in Appendix~\ref{app:CB-norm}
\end{proof}

%***************************************
\subsection{Special case of $\alpha=2$}
The case of $\alpha=2$ is of particular interest for us since computing
$V_{\alpha}(A; B)$ and $W_\alpha(A|B)$ are easier in this case. So we focus on this special case here, and find equivalent expressions for $V_2, W_2$.

\begin{lemma}\label{lem:V-2-expression}
We have
$$V_2(A;B) = \inf_{\tau_B, \sigma_B} \bigg(   \tr\Big[ \big(\rho_A^{-1/2}\otimes \tau_B^{-1/2}\big) \rho_{AB} \big(\rho_A^{-1/2}\otimes \sigma_B^{-1/2}\big) \rho_{AB}\Big]- \tr\Big[  \tau_B^{-1/2}\rho_B\sigma_B^{-1/2}\rho_B    \Big] 
       \bigg)^{1/2}$$
and
\begin{align}
W_2(A|B) 
& = \inf_{\tau_B, \sigma_B}  \bigg(\tr\big[ \tau_B^{-1/2}\rho_{AB} \sigma_B^{-1/2} \rho_{{A}B}\big] - \frac{1}{d_{A}}\tr\big[\tau_B^{-1/2} \rho_B\sigma_B^{-1/2} \rho_B  \big]\bigg)^{1/2}.\label{eq:W-2-0}
\end{align}
\end{lemma}

\begin{proof}
We compute
\begin{align*}
V_2^2(A;B) & = \Big\|  \Gamma_{\rho_A}^{-1/2} (\rho_{BA}) - \rho_B\otimes \rho_A^{1/2}     \Big\|_{(1,2)}^2\\
& = \inf_{\tau_B, \sigma_B} \Big\|  \Gamma_{\rho_A}^{-1/2} (\tau_B^{-1/4}\rho_{BA} \sigma_B^{-1/4}) - \tau_B^{-1/4}\rho_B\sigma_B^{-1/4}\otimes \rho_A^{1/2}     \Big\|_{2}^2\\
& = \inf_{\tau_B, \sigma_B} \left(\tr\Big[ \big(\rho_A^{-1/2}\otimes \tau_B^{-1/2}\big) \rho_{AB} \big(\rho_A^{-1/2}\otimes \sigma_B^{-1/2}\big) \rho_{AB}\Big]+ \tr\Big[  \tau_B^{-1/2}\rho_B\sigma_B^{-1/2}\rho_B    \Big]\right. \\
& \qquad\qquad \left.-\tr\Big[    \tau_B^{-1/2} \rho_{AB}  \sigma_B^{-1/2} \big(I_A\otimes \rho_{B}\big)    \Big] -\tr\Big[    \tau_B^{-1/2} \big(I_A\otimes \rho_{B}\big)   \sigma_B^{-1/2} \rho_{AB}   \Big]\right)\\
& = \inf_{\tau_B, \sigma_B} \left(\tr\Big[ \big(\rho_A^{-1/2}\otimes \tau_B^{-1/2}\big) \rho_{AB} \big(\rho_A^{-1/2}\otimes \sigma_B^{-1/2}\big) \rho_{AB}\Big]+ \tr\Big[  \tau_B^{-1/2}\rho_B\sigma_B^{-1/2}\rho_B    \Big]\right. \\
& \qquad\qquad \left.-2\tr\Big[    \tau_B^{-1/2} \rho_{B}  \sigma_B^{-1/2}  \rho_{B}    \Big] \right)\\
& = \inf_{\tau_B, \sigma_B} \left(\tr\Big[ \big(\rho_A^{-1/2}\otimes \tau_B^{-1/2}\big) \rho_{AB} \big(\rho_A^{-1/2}\otimes \sigma_B^{-1/2}\big) \rho_{AB}\Big]- \tr\Big[  \tau_B^{-1/2}\rho_B\sigma_B^{-1/2}\rho_B    \Big]\right).
\end{align*}

The proof of the second expression is similar. 

%For $W_2(A|B)$ we have
%\begin{align}
%W_2^2(A|B) & = \inf_{\tau_B, \sigma_B} \big\|  \tau_B^{-1/4} \rho_{AB} \sigma_B^{-1/4} -  \frac{I_{A}}{d_{A}}\otimes \tau_B^{-1/4} \rho_{B} \sigma_B^{-1/4}   \big \|_2^2 \nonumber \\
%& = \inf_{\tau_B, \sigma_B}  \left(\tr\big[ \tau_B^{-1/2}\rho_{AB} \sigma_B^{-1/2} \rho_{{A}B}\big] + \tr\big[\frac{1}{d_{A}^2} I_{A} \otimes \tau_B^{-1/2} \rho_B\sigma_B^{-1/2} \rho_B    \big] \nonumber\right. \\
%& \qquad \qquad \left.- \tr\big[\tau_B^{-1/2}\rho_{AB} \sigma_B^{-1/2} \frac{I_{A}}{d_{A}}\otimes \rho_{B}\big] - \tr\big[\frac{I_{A}}{d_{A}}\otimes\tau_B^{-1/2}\rho_{B} \sigma_E^{-1/2} \rho_{{A}B}\big]\right) \nonumber \\
%& = \inf_{\tau_B, \sigma_B} \left( \tr\big[ \tau_B^{-1/2}\rho_{AB} \sigma_B^{-1/2} \rho_{{A}B}\big] - \frac{1}{d_{A}}\tr\big[\tau_B^{-1/2} \rho_B\sigma_B^{-1/2} \rho_B  \big]\right)\nonumber.
%\end{align}
\end{proof}

It is also instructive to write down $V_2(A; B)$ for classical distributions $p_{AB}$:
$$V_2(A;B) = \sum_{b} \Big(   \sum_a p(a) \big(p(b|a) - p(b)     \big)^{2}    \Big)^{1/2}.$$
Given any realization $b\in \mathcal B$, we can view $p_{b|A}$ as a random variable (a function of the random variable $A$ with $p_{b|A}(a) = p(b|a)$). We have $\mathbb{E}_A\left[p_{b|A}\right]=\sum_{a}p(a)p(b|a)=p(b)$. Thus, 
\begin{align*}
V_2(A;B)=\sum_{b}\sqrt{\text{Var}_A\left[p_{b|A}\right]}.
\end{align*}
%Using the concept of $f$-divergence, we show in Theorem \ref{Tsalis}  in Appendix \ref{appendix} that $V_2^2(A;B)$ is equal to the ``mutual Tsallis information of order two". 
Another characterization of $V_2(A; B)$ can be found using the Bayes' rule:
\begin{align}
V_2(A;B)&=\sum_{b}\sqrt{\sum_{a}p(a)\big(p(b|a)-p(b)\big)^2}\nonumber\\
&=\sum_{b}p(b)\sqrt{\sum_{a}p(a)\left(\frac{p(b|a)}{p(b)}-1\right)^2}\nonumber\\
&=\sum_{b}p(b)\sqrt{\sum_{a}p(a)\left(\frac{p(a|b)}{p(a)}-1\right)^2}\nonumber\\
&=\sum_{b}p(b)\sqrt{\sum_{a}\left(\frac{p^2(a|b)}{p(a)}-2p(a|b)+p(a)\right)}\nonumber\\
&=\sum_{b}p(b)\sqrt{\sum_{a}\frac{p^2(a|b)}{p(a)}-1}.\nonumber
\end{align}
Thus,
\begin{align}
V_2(A;B)&=\sum_{b}p(b)\sqrt{\sum_{a}\frac{p^2(a|b)}{p(a)}-1}=\mathbb{E}_B\Big[\sqrt{\chi^2\left(p_{A|B}\parallel p_A\right)}\Big]\label{eqn:chi-sq1},
\end{align}
where $\chi^2(\cdot\|\cdot)$ is the $\chi$-square distance. The above formula has some interesting consequences:
\begin{itemize}
\item[(i)] Note that
\begin{align}
2^{\frac12 I_2(A;B)}&=\sum_{b}\sqrt{\sum_{a} p(a) p^2(b|a)}\nonumber\\
&=\sum_{b}\sqrt{\sum_{a}\frac{p^2(a,b)}{p(a)}}\nonumber\\
&=\sum_{b}p(b)\sqrt{\sum_{a}\frac{p^2(a|b)}{p(a)}}.\label{eqn:N5}
\end{align}
Comparing \eqref{eqn:chi-sq1} and \eqref{eqn:N5}, and utilizing the inequality 
$\sqrt{x}\geq \sqrt{x-1}\geq\sqrt{x}-1$ for $x\geq 1$, we obtain that  
\begin{align}
2^{\frac12 I_2(A;B)}\geq V_2(A;B)\geq 2^{\frac12 I_2(A;B)}-1.\label{eq:v_RI2_relation-new}
\end{align}
The above inequality is stronger than the one given in Proposition \ref{prop2} for $\alpha=2$ in the classical case.

\item[(ii)] Using the above expressions, proving the property of the monotonicity under local operations (Theorem~\ref{thm:monotonicity}) would be easier. 
For example, since 
the $\chi$-square distance retains monotonicity under local operations (the data processing inequality)~\cite{sason2016f},  we conclude that $V_2(X; B) \leq V_2(A; B)$.

%for every $b$, \emph{i.e.,}
%\begin{align*}
%\chi^2\left(p(a'|b)\parallel p(a')\right)&\leq\chi^2\left(p(a|b)\parallel p(a)\right).
%\end{align*}
%As a result one obtains another proof for the data processing inequality:
%\begin{align*}
%V_2(A';B)&=\mathbb{E}_B\sqrt{\chi^2\left(p(a'|B)\parallel p(a')\right)}\\
%&\leq\mathbb{E}_B\sqrt{\chi^2\left(p(a|B)\parallel p(a)\right)}=V_2(A;B).
%\end{align*}
\item[(iii)] When the marginal distribution $p_A$ is uniform over $\mathcal{A}$, we have
\begin{align*}
V_2(A; B)&=\mathbb{E}_B\Big[\sqrt{\chi^2\left(p_{A|B}\parallel p_A\right)}\,\Big]\\
&\leq\mathbb{E}_B\Bigg[\frac{\|p_{A|B}-p_A\|_{1}}{\sqrt{\frac{2}{|\mathcal{A}|}}}\Bigg]\\
&=\lVert p_{AB}-p_Ap_B\rVert_{1}\cdot\sqrt{|\mathcal{A}|/2},
\end{align*}
where for the inequality we use equation~(25) of \cite{sason2015upper}. This can be taught as a converse of Proposition~\ref{prop:non-decreasing}.  
\end{itemize}

Finally, another characterization of $V_2(A;B)$ for classical systems is given in Appendix~\ref{app:MI} where it is shown in Theorem~\ref{Tsalis}  that  $V^2_2(A;B)$ equals a \emph{Tsallis mutual information} of order two.

%*****************************************************************

\section{A decoupling theorem}\label{sec:decoupling}

Our main motivation for defining $V_\alpha$ is in its applications in decoupling type theorems. To explain this let us for example, think of the average of the so called \emph{purity} of $\rho_{A_0} = \tr_C(U \rho_A U^\dagger)$, i.e., $\E_U[\tr(\rho_{A_0}^2)]$, where the quantum system $A$ is composed of two subsystems $A_0, C$ and $U_A\in \fL(A)$ is a random unitary  distributed according to the Haar measure. Computing this average (using techniques that will be explained below) the result would be a multiple of $\tr(\rho_A^2)$ \emph{plus} a constant. Thus $\E_U[\tr(\rho_{A_0}^2)]$ cannot be naturally bounded by $c\tr(\rho_A^2)$ for some constant $c<1$. We conclude that for this problem it is more natural to replace purity with purity plus an appropriate constant. This simple modification is exactly what we do in using $V_\alpha(A; B)$ and $W_\alpha(A|B)$ instead of $I_{\alpha}(A; B)$ and $H_\alpha(A|B)$. 
The statement and the proof of the following decoupling theorem will clear up our point here.

In the following, we use (say) $A'$ to denote a copy of the system $A$. That is, $\cH_{A'}$ is a Hilbert space isomorphic to $\cH_A$, and $\mathcal A'=\mathcal A$ as sets. Let 
$$F_{AA'}: \cH_A\otimes \cH_{A'}\rightarrow \cH_A\otimes \cH_{A'},$$
to be the \emph{swap operator} given by
\begin{align}\label{eq:swap-op}
F_{AA'} \ket\psi_A\otimes \ket\varphi_{A'} = \ket\varphi_{A}\otimes \ket{\psi}_{A'}.\end{align}
Observe that $F_{AA'}^2= I_{AA'}$ and $\tr(F_{AA'}) = d_A$.

\begin{theorem}\label{thm:main-0}
Let $\rho_{AB}$ be an arbitrary quantum state and $\Phi:\fL(A)\rightarrow \fL(A_0)$ be an arbitrary completely positive map (not necessarily trace preserving) satisfying
$$\Phi\Big(\frac{I_A}{d_A}\Big) = \frac{I_{A_0}}{d_{A_0}}.$$
For a given unitary $U_A\in \fL(\cH_A)$ define 
$$\rho_{A_0B} =\Phi_A\otimes \mathcal I_B(U_A \rho_{AB}U_A^{\dagger}),$$
that is not necessarily normalized. 
Then for every $1\leq \alpha\leq 2$ the followings hold:

\begin{itemize}
\item[{\rm (i)}] We have 
$$\E_U\big[W_\alpha(A_0|B)\big] \leq 2^{\frac 2\alpha-1}\Big(\frac{\gamma-d_A/d_{A_0}}{d_A^2-1}\Big)^{\frac{1}{\alpha'}} W_\alpha(A|B),$$
where the expectation is taken with respect to the Haar measure and 
$$\gamma=\tr\big(F_{A_0A_0'} \Phi^{\otimes 2}(F_{AA'})\big).$$ 

\item[{\rm (ii)}]  Suppose that $\rho_A= I_A/d_A$ is maximally mixed. Then we have
$$\E_{U}  \big[V_\alpha(A_0;B)\big] \leq  2^{\frac 2\alpha-1} \big( \frac{d_{A_0}}{d_A} \big)^{\frac{1}{\alpha'}}\Big(\frac{\gamma-d_A/d_{A_0}}{d_A^2-1}\Big)^{\frac{1}{\alpha'}} V_\alpha(A;B).$$

\end{itemize}

\end{theorem}

This theorem in the special case of $\alpha=2$ (together with Proposition~\ref{prop:non-decreasing}) resembles the one-shot decoupling theorem of~\cite{DBWR14} with similar proof ideas. See also~\cite{Sharma15} for a similar decoupling type theorem.

The following corollary presents two important especial cases of this theorem.

\begin{corollary}\label{cor:main-1}
For an arbitrary quantum state $\rho_{AB}$ and $1\leq \alpha\leq 2$ the followings hold:

\begin{itemize}
\item[{\rm (a)}] If $A$ is  composed of two subsystems $A_0, C$ and for a unitary $U_A$ we define $\rho_{A_0B} = \tr_C\big( (U_A\otimes I_B) \rho_{AB} (U_A^\dagger\otimes I_B)   \big)$ then
$$\E_U\big[W_\alpha(A_0|B)\big] \leq 2^{\frac 2\alpha-1}d_C^{-\frac{1}{\alpha'}} W_\alpha(A|B),$$
where the expectation is taken with respect to the Haar measure. Moreover, if $\rho_A= I_A/d_A$ then 
$$\E_{U}  \big[V_\alpha(A_0;B)\big] \leq  2^{\frac 2\alpha-1}d_C^{-\frac{2}{\alpha'}} V_\alpha(A;B).$$

\item[{\rm (b)}]  Suppose that $\cH_{A_0}\subseteq \cH_A$ is a subspace and $P:\cH_{A} \rightarrow \cH_{A_0}$ is the orthogonal projection onto this subspace. Then for a unitary $U_A$ defining 
$$\rho_{A_0B} = \frac{d_A}{d_{A_0}} (P\otimes I_B) \rho_{AB} (P_A\otimes I_B),$$
we have
$$\E_U\big[W_\alpha(A_0|B)\big]\leq 2^{\frac 2\alpha-1}W_\alpha(A|B).$$ Moreover, if $\rho_A= I_A/d_A$ then 
$$\E_{U}  \big[V_\alpha(A_0;B)\big] \leq  2^{\frac 2\alpha-1}\big(\frac{d_{A_0}}{d_A}\big)^{\frac{1}{\alpha'}} V_\alpha(A;B).$$

\end{itemize}

\end{corollary}

Part (b) of this corollary gives the following generalization of the decoupling result of~\cite{Berta08}. To prove this corollary use part (b) of the above corollary together with Proposition~\ref{prop:non-decreasing}.  

\begin{corollary}\label{cor:decoupling-Berta}
Let $\rho_{AB}$ be bipartite quantum state and let $P:\cH_A\rightarrow \cH_{A_0}$ be an orthonormal projection. Then we have
$$\E_U\Big[  \Big\| \frac{d_A}{d_{A_0}} (P\otimes I_B) U_A \rho_{AB}U_A^\dagger (P\otimes I_B) - \frac{I_{A_0}}{d_{A_0}}\otimes\rho_B             \Big\|_1    \Big]\leq 2^{\frac 2\alpha-1} d_{A_0}^{\frac{1}{\alpha'}} W_\alpha(A|B)$$
\end{corollary}

Before getting into the proof of Theorem~\ref{thm:main-0} let us explain the classical counterpart of this theorem in which $A$ denotes a classical system. Due to its applications, we present only the classical counterpart of part (a) of Corollary~\ref{cor:main-1}.

\begin{theorem}\label{thm:main-3}
Let $\mathcal A=\mathcal A_0\times \mathcal C$ be arbitrary sets, and let
$$\rho_{AB}=\sum_{a}p(a)\ket a\bra a\otimes \rho_a,$$
be an arbitrary classical-quantum state. For a function $f: \mathcal A\rightarrow \mathcal A_0$ define
$$\rho_{A_0B} = \sum_{a} p(a) \ket{f(a)}\bra{f(a)}\otimes \rho_a.$$
Then for every $1\leq \alpha\leq 2$ the followings hold: 
\begin{itemize}
\item[{\rm (i)}] We have 
$$\E_f\big[W_\alpha(A_0|B)\big] \leq 2^{\frac 2\alpha-1} W_\alpha(A|B),$$
where the expectation is taken with respect to the uniform distribution over all $|\mathcal C|$-to-$1$ functions\footnote{A function $f$ is $k$-to-$1$ if $|f^{-1}(a_0)|=k$ for all $a_0$.} $f:\mathcal A\rightarrow \mathcal A_0$.

\item[{\rm (ii)}]  Suppose that $p(a)= 1/|\mathcal A|$ is the uniform distribution. Then we have
$$\E_{f}  \big[V_\alpha(A_0;B)\big] \leq  2^{\frac 2\alpha-1} |\mathcal C|^{-\frac{1}{\alpha'}} V_\alpha(A;B),$$
where the expectation is taken with respect to the uniform distribution over all $|\mathcal C|$-to-$1$ functions $f:\mathcal A\rightarrow \mathcal A_0$.

\end{itemize}

\end{theorem}

To prove the above theorems we first use the Riesz-Thorin theorem to reduce the statement for a general $1\leq \alpha\leq 2$ to the special cases of $\alpha=1$ and $\alpha=2$. The proof for $\alpha=1$ follows from a simple application of the triangle inequality. To prove the theorem for $\alpha=2$ we need to compute certain averages over a Haar random unitary (random permutation). In the following, we first explain some tools for computing these averages and then present the proof of the above theorems.

\begin{lemma}\cite{DBWR14} \label{lem:swap-trace} 
For any $M_A\otimes N_{A'} \in \fL(\cH_A\otimes \cH_{A'})$  we have
$$\tr[F_{AA'}(M\otimes N)]= \tr[MN],$$
where $F_{AA'}$ is the swap operator defined by~\eqref{eq:swap-op}. 
\end{lemma}

\begin{lemma}\cite{DBWR14}\label{lem:Haar-expectation} 
Let $M_{AA'} \in \fL(\cH_A\otimes \cH_{A'})$. Then we have
$$\E_U\big[  (U\otimes U) M_{AA'} (U^\dagger\otimes U^\dagger)     \big] = \alpha I_{AA'} + \beta F_{AA'},$$
where the expectation is taken with respect to the Haar measure and $\alpha, \beta$ are determined by
\begin{align*}
\tr[M] &= \alpha d_A^2+\beta d_A,\\
\tr[MF] &= \alpha d_A + \beta d_A^2.
\end{align*}
\end{lemma}

\begin{corollary}\label{cor:Haar-expectation-2}
Let $M_{AA'BB'} \in \fL(\cH_{AB}\otimes \cH_{A'B'})$. Then we have
\begin{align*}
\E_{U_A}\big[  (U_A\otimes U_{A'}) M_{AA'BB'} (U_A^\dagger\otimes U_{A'}^\dagger)     \big] = \frac{1}{d_A^2-1}\Big[ &I_{AA'}\otimes \tr_{AA'}(M) - \frac{1}{d_A}I_{AA'}\otimes \tr_{AA'}(F_{AA'} M)  \\
&~+ F_{AA'}\otimes \tr(F_{AA'} M) -\frac{1}{d_A} F_{AA'}\otimes \tr_{AA'}(M)     \Big],
\end{align*}
\end{corollary}

\begin{proof}
To simplify the expressions let us denote $d=d_A.$
Decompose $M$ as
$$M_{AA'BB'} = \sum_j (X_j)_{AA'}\otimes (Y_j)_{BB'}.$$
Define
\begin{align*}
\alpha_j &= \frac{1}{d^2-1} \tr(X_j) - \frac{1}{d(d^2-1)} \tr(F_{AA'}X_j)\\
\beta_j &= \frac{1}{d^2-1} \tr(F_{AA'}X_j) - \frac{1}{d(d^2-1)} \tr(X_j).
\end{align*}
Note that $\alpha_j, \beta_j$ satisfy 
\begin{align*}
\tr[X_j] &= \alpha_j d^2+\beta_j d,\\
\tr[F_{AA'}X_j] &= \alpha_j d + \beta_j d^2.
\end{align*}
Thus by Lemma~\ref{lem:Haar-expectation} we have
\begin{align*}
\E_{U}\big[  U^{\otimes 2} M (U^\dagger)^{\otimes 2}    \big] & =\sum_j  \E_{U}\big[  U^{\otimes 2} X_j (U^\dagger)^{\otimes 2}    \big]\otimes Y_j\\
& = \sum_j \big(\alpha_j I_{AA'} + \beta_j F_{AA'}\big)\otimes  Y_j.
\end{align*}
Then the desired result follows once we note that 
\begin{align*}
\sum_j \alpha_j Y_j &= \frac{1}{d^2-1} \tr_{AA'}(M) - \frac{1}{d(d^2-1)}\tr_{AA'}(F_{AA'} M),\\
\sum_j \beta_j Y_j &= \frac{1}{d^2-1} \tr_{AA'}(F_{AA'}M) - \frac{1}{d(d^2-1)}\tr_{AA'}(M).
\end{align*}

\end{proof}

\begin{proof}[Proof of Theorem \ref{thm:main-0}]
The proof of part (ii) is immediate once we have (i). The point is that when $\rho_A = I/d_A$, then $\rho_{A_0} = I/d_{A_0}$. In this case we have
$$V_{\alpha}(A; B) = \Big\| \Gamma_{I/d_A}^{-1/\alpha'} \big( \rho_{BA} - \rho_B\otimes I/d_A    \big)  \Big\|_{(1, \alpha)} = d_A^{\frac{1}{\alpha'}} \big\|\rho_{BA} - \rho_B\otimes I/d_A \big\|_{(1, \alpha)} = d_A^{\frac{1}{\alpha'}} W_\alpha(A|B),$$
and similarly $V_\alpha(A_0; B) = d_{A_0}^{\frac{1}{\alpha'}} W_\alpha(A_0|B)$. Using these in (i), part (ii) will be implied. So we focus on the proof of (i). 

Let  $\mathcal U_A\subset \fL(A)$ be the space of unitary operators (equipped with the Haar measure). Define $\Xi: \fL(BA)\rightarrow \ell(\mathcal U_A)\otimes\fL(BA_0)$ by
$$\Xi(M_{AB})(U_A) := \Phi_{A\rightarrow A_0}\big(  U_A M_{BA} U_A^\dagger    \big) - \tr_A(M_{AB}) \otimes \frac{I_{A_0}}{d_{A_0}}.$$
Suppose that for every $1\leq \alpha\leq 2$ we have 
\begin{align}\label{eq:Xi-norm-1}
\|\Xi\|_{(1, \alpha)\rightarrow (1, 1, \alpha)} \leq 2^{\frac 2\alpha-1}\Big(\frac{\gamma-d_A/d_{A_0}}{d_A^2-1}\Big)^{-\frac{1}{\alpha'}}.
\end{align}
That is, for every $M_{AB}$ we have
\begin{align}\label{eq:Xi-norm-2}
\E_{U_A} \Big[\Big\|\Phi_{A\rightarrow A_0}\big(  U_A M_{BA} U_A^\dagger    \big) - \tr_A(M_{BA}) \otimes \frac{I_{A_0}}{d_{A_0}}\Big\|_{(1, \alpha)}\Big] \leq 2^{\frac 2\alpha-1}\Big(\frac{\gamma-d_A/d_{A_0}}{d_A^2-1}\Big)^{-\frac{1}{\alpha'}}\|M_{BA}\|_{(1, \alpha)}.
\end{align}
Then part (i) follows once in the above inequality we put $M_{AB} = \rho_{AB} - \rho_B\otimes I/d_A$. So we just need to prove~\eqref{eq:Xi-norm-1}. Now the point is that the $(1, \alpha)$-norms as well as $(1, 1, \alpha)$-norms for $1\leq \alpha\leq 2$ form an interpolation family~\cite{Pisier}. Thus by the Riesz-Thorin theorem (see Appendix~\ref{app:interpolation}) proving~\eqref{eq:Xi-norm-1} for values of $\alpha=1$ and $\alpha=2$ implies it for all $1\leq \alpha\leq 2$. So in the following, we focus on the proof of~\eqref{eq:Xi-norm-2} for special cases of $\alpha=1$ and $\alpha=2$.

First let $\alpha=1$. Let us writte $\Xi= \Xi_0- \Xi_1$ where 
$\Xi_0(M_{BA})(U_A) = \Phi_{A\rightarrow A_0}\big(U_A M_{BA} U_A^\dagger\big)$ and $\Xi_0(M_{BA})(U_A) = \tr_A(M_{BA})\otimes I_{A_0}/d_{A_0}$.
Then by the triangle inequality we have
$$\big\|\Xi\big\|_{1\rightarrow 1}\leq \big\|\Xi_0\big\|_{1\rightarrow 1} + \big\|\Xi_1\big\|_{1\rightarrow 1}.$$
So it suffices to show that each term on right hand side is at most $1$.
That is, we need to show that for every $M_{AB}$ we have
$$\big\|\Xi_j(M_{BA})\big\|_{1} \leq \|M_{BA}\|_1, \qquad j=0, 1.$$
Since $\Xi_j$ for $j=0, 1$ are completely positive, by~\cite[Corollary 6]{DJKR}, it suffices to prove the above inequality for $M_{BA}\geq 0$ positive semidefinite. For $j=0$ we have
\begin{align*}
\big\|\Xi_0(M_{BA})\big\|_{1}& =\E_U\Big[\Big\|\Phi_{A\rightarrow A_0}\big(  U_A M_{BA} U_A^\dagger    \big)\Big\|_{1}\Big]\\
 & = \E_U\Big[\tr\Big(\Phi_{A\rightarrow A_0}\big(  U_A M_{BA} U_A^\dagger    \big)\Big)\Big]\\ 
&=  \tr\Big(\Phi_{A\rightarrow A_0}\E_U\big(  U_A M_{BA} U_A^\dagger    \big)\Big)\\
& = \tr(M_{AB}) \tr\Big(\Phi\Big(\frac{I_{AB}}{d_Ad_B}\Big)\Big)\\
& = \tr(M_{AB}) \tr\Big( \frac{I_{A_0B}}{d_{A_0}d_B} \Big)\\
& = \|M_{AB}\|_1.
\end{align*}
For $j=1$ we have
\begin{align*}
\big\|\Xi_1(M_{BA})\big\|_{1} & = \big\|\tr_A(M_{BA})\otimes I_{A_0}/d_{A_0}\big\|_1 = \big\|\tr_A(M_{BA})\big\|_1 = \tr(M_{BA}) = \|M_{BA}\|_1.
\end{align*}
We are done with the case $\alpha=1$.

Proof of~\eqref{eq:Xi-norm-2} for $\alpha=2$ needs more work. For given density matrices $\tau_B, \sigma_B$ define 
\begin{align}\label{eq:notation-hat-1}
\hat M_{AB} = \tau_B^{-1/4} M_{AB} \sigma_B^{-1/4}, \quad \hat M_B =\tr_A(\hat M_{AB}) = \tau_B^{-1/4}\tr_A(M_{AB}) \sigma_B^{-1/4},
\end{align}
and for a unitary $U_A$ define
$$M_{A_0B} = \Phi_{A\rightarrow A_0}(U_A M_{AB} U_A^\dagger), \qquad \hat M_{A_0B} = \tau_B^{-1/4} M_{A_0B}\sigma_{B}.$$
Following similar computations as in the proof of Lemma~\ref{lem:V-2-expression} we have
\begin{align*}
\Big\| M_{BA_0} - \tr_{A_0}(M_{BA_0}) &\otimes \frac{I_{A_0}}{d_{A_0}}\Big\|_{(1, 2)}^2 = \inf_{\tau_B, \sigma_B} \tr\Big[  \hat M_{BA_0}   \hat M_{BA_0}^\dagger      \Big] - \frac{1}{d_{A_0}} \tr\big(  \hat M_B \hat M_B^\dagger  \big),
\end{align*}
For fix $\tau_B, \sigma_B$, by Lemma~\ref{lem:swap-trace} we have
\begin{align}
\E_{U_A}\tr \big[ \tau_B^{-1/2}M_{{A_0}B} \sigma_B^{-1/2} M_{{A_0}B}^\dagger\big] & = \E_{U_A} \tr\big[ \Phi(U_A\hat M_{AB}U_A^{\dagger}) \cdot \Phi(U_A\hat M_{AB}^\dagger U_A^{\dagger}) \big]\nonumber\\
& = \E_{U_A} \tr\big[ F_{{A_0}B{A_0}'B'}~\Phi(U\hat M_{AB}U^{\dagger}) \otimes \Phi(U\hat M_{A'B'}^\dagger U^{\dagger}) \big]\nonumber\\
& = \E_{U_A} \tr\big[ \Phi^{*}\otimes \Phi^*(F_{{A_0}B{A_0}'B'})~U\hat M_{AB}U^{\dagger}\otimes U\hat M_{A'B'}^\dagger U^{\dagger} \big],
\label{eq:proof-01}
\end{align}
where $\Phi^*_{A_0\rightarrow A}$ is the adjoint of $\Phi$ with respect to the Hilbert-Schmidt inner product. 
Now using Corollary~\ref{cor:Haar-expectation-2} we compute
\begin{align*}
\E_{U}\big[U^{\otimes 2}\hat M_{AB}\otimes \hat M_{A'B'}^\dagger (U^\dagger)^{\otimes 2}\big] & = \frac{1}{d_A^2-1} \Big[
I_{AA'} \otimes \hat M_{B}\otimes\hat M_{B'}^\dagger -\frac{1}{d_A} I_{AA'}\otimes \hat \mu_{BB'}\nonumber\\
& \qquad \qquad \quad +F_{AA'}\otimes \hat \mu_{BB'}-\frac{1}{d_A} F_{AA'}\otimes \hat M_B\otimes \hat M_{B'}^\dagger
\Big],
\end{align*}
where 
\begin{align}\label{eq:hat-mu-03}
\hat \mu_{BB'}=\tr_{AA'}[F_{AA'} \hat M_{AB}\otimes \hat M_{A'B'}^\dagger].
\end{align}
Therefore,
\begin{align*}
\E_{U_A}\tr \big[ \tau_B^{-1/2}M_{{A_0}B} \sigma_B^{-1/2} M_{{A_0}B}\big] & = \frac{1}{d_A^2-1}\tr\bigg[
\Phi^*\otimes \Phi^*(F_{{A_0}BA_0'B'}) \bigg(I_{AA'} \otimes \hat M_{B}\otimes\hat M_{B'}^\dagger -\frac{1}{d_A} I_{AA'}\otimes \hat \mu_{BB'}\\
& \qquad \qquad \qquad\qquad +F_{AA'}\otimes \hat \mu_{BB'}-\frac{1}{d_A} F_{AA'}\otimes \hat M_B\otimes \hat M_{B'}^\dagger\bigg)\bigg]\\
& = \frac{1}{d_A^2-1}\tr\bigg[
F_{{A_0}BA_0'B'} \bigg(\Phi^{\otimes 2}(I_{AA'}) \otimes \hat M_{B}\otimes\hat M_{B'}^\dagger -\frac{1}{d_A} \Phi^{\otimes 2}(I_{AA'})\otimes \hat \mu_{BB'}\\
& \qquad \qquad \qquad\qquad +\Phi^{\otimes 2}(F_{AA'})\otimes \hat \mu_{BB'}-\frac{1}{d_A} \Phi^{\otimes 2}(F_{AA'})\otimes \hat M_B\otimes \hat M_{B'}^\dagger\bigg)\bigg]\\
& =\frac{1}{d_A^2-1}\bigg[ \frac{d_{A}^2}{d_{A_0}}\tr\big(\hat M_{B}\hat M_{B}^\dagger\big) -\frac{d_A}{d_{A_0}} \tr\big(\hat M_{AB}\hat M_{AB}^\dagger \big)\\
& \qquad \qquad \qquad + \gamma \tr\big(\hat M_{AB}\hat M_{AB}^\dagger \big) - \frac{\gamma}{d_A} \tr\big(\hat M_{B}\hat M_{B}^\dagger\big)
\bigg]\\
& = \frac{1}{d_A^2-1}\bigg[ \Big(\gamma-\frac{d_A}{d_{A_0}}\Big) \tr\big(\hat M_{AB}\hat M_{AB}^\dagger \big) + \Big(\frac{d_A^2}{d_{A_0}} - \frac{\gamma}{d_A}\Big)\tr\big(\hat M_{B}\hat M_{B}^\dagger\big)
\bigg].
\end{align*}
Therefore, using the convexity of the square function we have
\begin{align*}
\Big(\E_U \Big[\Big\| M_{BA_0} - \tr_{A_0}(M_{BA_0})& \otimes \frac{I_{A_0}}{d_{A_0}}\Big\|_{(1, 2)}\Big]\Big)^2  \leq \E_U\Big[\Big\| M_{BA_0} - \tr_{A_0}(M_{BA_0}) \otimes \frac{I_{A_0}}{d_{A_0}}\Big\|_{(1, 2)}^2\Big] \\
& = \frac{1}{d_A^2-1}\bigg[ \Big(\gamma-\frac{d_A}{d_{A_0}}\Big) \tr\big(\hat M_{AB}\hat M_{AB}^\dagger \big) + \Big(\frac{d_A^2}{d_{A_0}} - \frac{\gamma}{d_A}\Big)\tr\big(\hat M_{B}\hat M_{B}^\dagger\big)- \frac{1}{d_{A_0}} \tr\big(\hat M_{B}\hat M_{B}^\dagger\big)\\
& = \frac{\gamma-d_A/d_{A_0}}{d_A^2 -1}\bigg[ \tr\big(\hat M_{AB}\hat M_{AB}^\dagger \big) - \frac{1}{d_A}\tr\big(\hat M_{B}\hat M_{B}^\dagger\big)
\bigg]\\
& \leq \frac{\gamma-d_A/d_{A_0}}{d_A^2 -1}\tr\big(\hat M_{AB}\hat M_{AB}^\dagger \big).
\end{align*}
Taking infimum over the choice of $\tau_B, \sigma_{B}$ we find that
\begin{align*}
\E_U \Big[\Big\| M_{BA_0} - \tr_{A_0}(M_{BA_0}) \otimes \frac{I_{A_0}}{d_{A_0}}\Big\|_{(1, 2)}\Big] \leq \Big(\frac{\gamma-d_A/d_{A_0}}{d_A^2 -1}\Big)^{1/2} \Big\| M_{BA} \Big\|_{(1, 2)}.
\end{align*}

\end{proof}

\begin{proof}[Proof of Theorem \ref{thm:main-3}]
The proof is similar to that of Theorem~\ref{thm:main-0}. Again part (ii) is an immediate consequence of part (i). Also, for part (i) it suffices to prove that for every 
$$M_{AB} = \sum_a \ket a\bra a\otimes N_a,$$
the inequality 
$$\E_f\Bigg[   \Big\|M_{BA_0} - \tr_{A_0}(M_{BA_0})\otimes \frac{I_{A_0}}{|\mathcal A_0|}\Big\|_{(1, \alpha)}   \Bigg] \leq 2^{\frac 2 \alpha-1} \|M_{BA}\|_{(1, \alpha)},$$
holds for all $1\leq \alpha\leq 2$, where 
$$M_{A_0B} = \sum_a \ket{f(a)}\bra{f(a)}\otimes N_a,$$
and the average is with respect to the uniform distribution over all $|\mathcal C|$-to-$1$ functions $f$.  Moreover, by the Riesz-Thorin theorem it suffices to prove the above inequality for $\alpha=1$ and $\alpha=2$.

Again the proof for $\alpha=1$ is a simple consequence of the triangle inequality which we do not repeat.  For $\alpha=2$ we first use
$$\Bigg(\E_f\Bigg[   \Big\|M_{BA_0} - \tr_{A_0}(M_{BA_0})\otimes \frac{I_{A_0}}{|\mathcal A_0|}\Big\|_{(1, 2)}   \Bigg]\Bigg)^2 \leq \E_f\Bigg[   \Big\|M_{BA_0} - \tr_{A_0}(M_{BA_0})\otimes \frac{I_{A_0}}{|\mathcal A_0|}\Big\|_{(1, 2)}^2   \Bigg],$$
and then to estimate the left hand side we follow similar steps as in the proof of Theorem~\ref{thm:main-0}. We only need to replace the average with respect to the Haar measure with another average.

A uniformly random $|\mathcal C|$-to-$1$ function $f:\mathcal A\rightarrow \mathcal A_0$ can be chosen as follows: let $\pi$ be a uniformly random permutation on $\mathcal A=\mathcal A_0\times\mathcal C$. Then $f(a)=\pi_0(a)$ is a uniformly random $|\mathcal C|$-to-$1$ function where by $\pi_0(a)$ we mean the first coordinate of $\pi(a)\in \mathcal A=\mathcal A_0\times \mathcal C$. With this choice of $f$, the operator $M_{A_0B}$ can be written as
$$M_{A_0B} = \tr_C\Big(  (U_\pi\otimes I_B) M_{AB}(U_{\pi}^\dagger\otimes I_B)     \Big),$$
where $U_\pi$ is the permutation matrix associated with $\pi$. Now we can follow the proof of Theorem~\ref{thm:main-0}. Fixing $\sigma_B, \tau_B$, using notations in~\eqref{eq:notation-hat-1} and replacing $U$ with $U_\pi$, equation~\eqref{eq:proof-01} is still valid for the choice of $\Phi$ being the partial trace with respect to $C$. Nevertheless, instead of Corollary~\ref{cor:Haar-expectation-2} we should use
\begin{align}
\E_{\pi}\big[U_\pi^{\otimes 2}\hat M_{AB}\otimes \hat M_{A'B'}^\dagger (U_\pi^\dagger)^{\otimes 2}\big] & = \frac{1}{|\mathcal A|^2 - |\mathcal A|} I_{AA'} \otimes \Big( \hat{M}_B\otimes \hat{M}_{B'}^{\dagger} - \hat{\mu}_{BB'}  \Big) \nonumber\\
&\qquad + J_{AA'}\otimes \bigg( \Big(\frac{1}{|\mathcal A|^2 - |\mathcal A|} +\frac{1}{|\mathcal A|}\Big)\hat \mu_{BB'}   - \frac{1}{|\mathcal A|^2 -|\mathcal A|} \hat M_B\otimes \hat M_{B'}^\dagger   \bigg),\label{eq:pi-average-0}
\end{align}
where $\hat\mu_{BB'}$ is given by~\eqref{eq:hat-mu-03} and 
$$J_{AA'} = \sum_{a} \ket {a}\bra a\otimes \ket a\bra a. $$
This equation can be proven using 
\begin{align*}
\E_\pi\Big[ U_\pi^{\otimes 2} \ket a\bra a\otimes \ket {a'}\bra{a'}U_\pi^{\otimes 2}\Big] & = \begin{cases}
\frac{1}{|\mathcal A|^2 - |\mathcal A|} (I_{AA'} - J_{AA'}) & a\neq a', 
\\
\frac{1}{|\mathcal A|} J_{AA'} & a=a'.
\end{cases}
\end{align*}
Then the proof follows by putting~\eqref{eq:pi-average-0} in~\eqref{eq:proof-01}, using $F_{AA'}J_{AA'} =J_{AA'}$, $\tr J_{AA'} = |\mathcal A|$, and a straightforward computation.

\end{proof}

The ratio $|\mathcal{C}|^{-\frac{1}{\alpha'}}$  in Theorem \ref{thm:main-3} is asymptotically tight up to a constant as the following example shows:
\begin{example} \label{exmple1} 
Let $p_A$ be uniform on $\mathcal{A}$, and $p_{B|A}$ be a classical erasure channel, \emph{i.e.,} the alphabet set of $B$ is $\mathcal B=\{\mathsf e\}\cup\mathcal{A}$, and for all $a\in \mathcal A$,
$$p_{B|A}(\mathsf e|a)=\epsilon, \qquad p_{B|A}(a|a)=1-\epsilon,$$
and $p_{B|A}(a'|a)=0$ if $a'\neq a$.
Then a direct calculation shows that
\begin{align*}
W_\alpha(A|B)&=(1-\epsilon)\left[\Big(1-\frac{1}{|\mathcal A|}\Big)^{\alpha} + \Big(|\mathcal A|-1\Big)\frac{1}{|\mathcal A|^\alpha}
\right]^{\frac{1}{\alpha}}
\\
V_\alpha(A;B)&=|\mathcal{A}|^{\frac{1}{\alpha'}}\cdot W_\alpha(A|B).
\end{align*}
Furthermore, for any $|\mathcal C|$-to-1 function $f:\mathcal A\rightarrow \mathcal A_0$ with
$A_0=f(A)$ we have
\begin{align*}
W_\alpha(A_0|B)&=(1-\epsilon)\left[\Big(1-\frac{1}{|\mathcal A_0|}\Big)^{\alpha} + \Big(|\mathcal A_0|-1\Big)\frac{1}{|\mathcal A_0|^\alpha}
\right]^{\frac{1}{\alpha}}
\\
V_\alpha(A_0;B)&=|\mathcal{A}_0|^{\frac{1}{\alpha'}}\cdot W_\alpha(A_0|B)
\end{align*}
Hence, for fixed $|\mathcal{C}|=\frac{|\mathcal A|}{|\mathcal A_0|}$ when $|\mathcal{A}|$ tends to infinity we have
$$\lim_{|\mathcal{A}|\rightarrow \infty} \min_{f: \text{\rm{$|\mathcal{C}|$-to-1}}}\frac{W_\alpha(A_0|B)}{W_\alpha(A|B)}=1,$$
$$\lim_{|\mathcal{A}|\rightarrow \infty} \min_{f: \text{\rm{$|\mathcal{C}|$-to-1}}}\frac{V_\alpha(A_0 ;B)}{V_\alpha(A;B)}=|\mathcal{C}|^{-\frac{1}{\alpha'}},$$
for any $\alpha>1$. Thus Theorem~\ref{thm:main-3} is asymptotically tight up to a constant.
\end{example}

%***********************************
\section{Applications in secrecy} \label{sec:app}

The common practice in the information theoretic security literature is to use mutual information and conditional entropy to measure the amount of leakage to an adversary. In particular, for a message $A$ and adversary's side information $B$, the conditional entropy $H(A|B)=H(A)-I(A;B)$, called the \emph{equivocation}, is the most favorite measure. When $I(A;B)$ is small, or equivalently $H(A|B)$ is close to $H(A)$, by Pinsker's inequality\footnote{Here we use $I(A;B)= D(\rho_{AB}\| \rho_A\otimes \rho_B)$.} the trace distance between $\rho_{AB}$ and $\rho_A\otimes \rho_B$ is small too. Nevertheless as will be shown later in this section, mutual information has some disadvantages as a secrecy parameter. 

 Here we suggest the use of $V_\alpha(A;B)$ or $I_\alpha(A;B)$  as a replacement of mutual information for measuring secrecy.\footnote{Note that by Proposition~\ref{prop2} the $\alpha$-R\'enyi mutual information and $V_\alpha(A; B)$ are related quantities.} The point is that, by Proposition~\ref{prop:non-decreasing} when $V_\alpha(A;B)$ is small, again $\rho_{AB}$ and $\rho_A\otimes \rho_B$ are close in trace distance. Moreover, our decoupling theorems in the previous section can be used to prove more effective exponentially small bounds on $V_\alpha$.   

There have been a few recent works that provide further justifications for our suggestion. Controlling the R\'enyi mutual information of order infinity $I_\infty(A;B)$ finds an operational justification in~\cite{Kamath}. Since the R\'enyi mutual information is non-decreasing as a function of its order, any upper bound on R\'enyi mutual information of order infinity yields a bound on R\'enyi mutual information of other orders. 
Moreover, authors in~\cite{yu2017r} study the secure capacity of the wiretap channel when the security is measured by the R\'enyi divergence.\footnote{ There are also other approaches for defining security metrics (e.g. see \cite{Issa, Kamath, Li, Cuff,Weinberger,bellare2012semantic,dodis2005entropic}) based on different correlation measures. 
}

In the following, we study the problem of \emph{privacy amplification} and present an application of our new correlation measure and decoupling theorems there. Moreover, we discuss the advantages of $V_\alpha$ as a secrecy parameter over mutual information in this problem. Also, in Appendix~\ref{app:semantic-sec} we show a connection between $V_\alpha$ for $\alpha=\infty$ and \emph{semantic security}.

%**********************
\subsection{Privacy amplification}
Suppose that a party has a secret key $A$ of $k$ uniform random bits, i.e., $\mathcal A=\{0,1\}^k$ and $p_A$ is the uniform distribution over $\mathcal A$. However, the key has partially leaked to an adversary who has access to a quantum register $B$ which is correlated with the secret key $A$ according to some known $\rho_{AB}$ with
$$\rho_{AB} = \frac{1}{2^k}\sum_{a\in \{0,1\}^k} \ket a\bra a\otimes \rho_a.$$
 Level of security of the key depends on the amount of information obtainable by the eavesdropper and may be measured by a correlation metric between the secret key $A$ and adversary's subsystem~$B$.

Suppose that we want to decrease the correlation between $B$ and the secret key at the cost of reducing the length of the key (privacy amplification). More precisely, suppose that we have a  function $f: \mathcal A=\{0,1\}^{k}\rightarrow \mathcal A_0=\{0,1\}^{k-1}$ such that $A_0=f(A)$ is uniform over $\{0,1\}^{k-1}$. Then by replacing $A$ with $A_0=f(A)$, and reducing the number of bits in the key, we expect to reduce the amount of correlation between the key and $B$. Indeed, 
we are interested in finding a suitable function $f$ such that the correlation between the distilled secret key $A_0$ and $B$ is minimized. 

Measuring the correlation between the key and $B$ in terms of $V_\alpha$ for $\alpha\in (1,2]$ and using Theorem~\ref{thm:main-3}, if we are willing to reduce the length of key from $k$ to $k-\ell$ bits, there exists a  $2^\ell$-to-1 function $f:\{0,1\}^k\rightarrow \{0,1\}^{k-\ell}$ such that 
\begin{align}\label{eqn:ReN}
V_\alpha(A_0;B) \leq 
2^{\frac 2\alpha-1} 2^{-\frac{\ell}{\alpha'}}V_\alpha(A;B),
\end{align}
where $A_0=f(A)$ is still uniform. Therefore, the correlation between the key and $B$ reduces exponentially in $\ell$. Observe that from Proposition~\ref{prop:non-decreasing} and Proposition~\ref{prop2} we obtain that there exists a  $2^\ell$-to-1 function such that for $A_0=f(A)$ we have
\begin{align}
\Big \| \rho_{A_0B}-\frac{I_{A_0}}{2^{k-\ell}}\otimes \rho_B    \Big\|_{1} \leq 2^{-\frac{\ell}{\alpha'}} \big(2^{\frac{1}{\alpha'} I_\alpha(A; B)}+1\big).
\label{eqn:ReN2}
\end{align}
When $\alpha=2$, using~\eqref{eq:v_RI2_relation-new} the above bound can be improved to  
\begin{align}
\Big \| \rho_{A_0B}-\frac{I_{A_0}}{2^{k-\ell}}\otimes \rho_B    \Big\|_{1} \leq 2^{-\ell/2} 2^{\frac12 I_2(A;B) }
\label{eqn:ReN3}
\end{align}
when $B$ is classical. 

Inequality~\eqref{eqn:ReN3} can also be obtained by the result of Renner on privacy amplification \cite[Theorem 5.5.1]{Renner} for classical-quantum systems (see also~\cite{RK05}).  Nevertheless,~\eqref{eqn:ReN} is stronger than Renner's result, at least in the fully classical case. While Renner's result works only for $\alpha=2$, equation~\eqref{eqn:ReN} allows for all orders $\alpha\in(1,2]$. On the other hand, Renner's result is more general because it does not assume uniform distribution on the random variable $A$. 

A closely related result is in Hayashi's work on privacy amplification~\cite[Theorem 1]{hayashi2011exponential}. Even though this result is stated in terms of mutual information, the key step in its proof is the following theorem (see equation~(29) of~\cite{hayashi2011exponential}). This theorem should be compared with part (i) of Theorem~\ref{thm:main-3}.

\begin{theorem}[\cite{hayashi2011exponential}] 
Let $\mathcal A=\mathcal A_0\times \mathcal C$ and $\mathcal{B}$ be arbitrary finite sets, and let $p_{AB}$ be an arbitrary bipartite distribution. Then for any $\alpha\in(1,2]$ we have
\begin{align}
\mathbb{E}_f \left[2^{-\tilde H_{\alpha}(A_0|B)}\right]\leq 2^{-\tilde H_\alpha(A|B)}+\frac{1}{|\mathcal{A}_0|^{\alpha-1}},
\label{gdferd}
\end{align}
where $A_0=f(A)$ and the expectation is taken with respect to the uniform distribution over all $|\mathcal C|$-to-$1$ functions $f$ (or over a class of two-universal $|\mathcal C|$-to-$1$ hash functions). Here, the following definition of the conditional R\'enyi entropy is utilized:
$$\tilde H_{\alpha}(A|B)=-\log \sum_{a,b} p_{B}(b)p_{A|B}(a|b)^\alpha.$$
\end{theorem}

Hayashi uses a different definition of conditional R\'enyi entropy than the one used in this paper; Comparing to~\eqref{eqn-cond-entr-ren} there is no minimization in the definition of~$\tilde H_\alpha(A|B)$.  Furthermore, our theorem does not have an additive term like $\frac{1}{|\mathcal{A}_0|^{\alpha-1}}$ as in \eqref{gdferd}. 
We should also remark that, in our results, similar to Hayashi's, the uniform distribution over all $|\mathcal C|$-to-$1$ functions can be replaced with the uniform distribution over a class of two-universal hash functions simply because in the proofs we only use the first and second moments of the underlying distribution on the functions.

 %******************
\subsection{Mutual information versus $V_\alpha$}
We mentioned above that Shannon used mutual information as a secrecy parameter, while we propose to use $V_\alpha$ for $\alpha\in(1,2]$ instead. Here we discuss this in more details.  
Let us start with a result similar to Theorem~\ref{thm:main-3} for mutual information.

\begin{theorem}
\label{theorem_Shearer}
Let $\mathcal A=\{0,1\}^{k}$, and let $p_{AB}$ be such that $p_A$ is the uniform distribution over $\{0,1\}^{k}$. Then there exists a $2$-to-$1$ function $f: \{0,1\}^{k}\rightarrow \{0,1\}^{k-1}$ such that for $A_0=f(A)$ we have
$$I(A_0; B) \leq \frac{k-1}{k}\,I(A; B).$$
Furthermore, the ratio $(k-1)/k$ in the above statement is optimal and cannot be replaced with a smaller constant that depends only on $k$ (and not on $p_{AB}$).
\end{theorem}

\begin{proof}
Let us denote the $i$-th bit of $A$ by $A_i$, so that $A=(A_1, \dots, A_k)$. We let $f$ to be the function that drops one bit of $A$. Indeed, we let $A_0=A_S=f_{S}(A)$ where $S$ is some $(k-1)$-element subset of $\{1, \dots, k\}$, and $A_S$ is the subsequence of its associated bits. 
By Shearer's lemma~\cite{chung1986some} we have 
\begin{align*}
\frac{1}{k}\sum_{S:\, \lvert S\rvert=k-1}H(A_S|B)&\geq \frac{k-1}{k}H(A|B).
\end{align*} 
Since $I(A;B)=k-H(A|B)$ and $I(A_S;B)=(k-1)-H(A_S|B)$ for any subset $S$ of size $k-1$, we obtain
\begin{align*}
\frac{1}{k}\sum_{S:\, \lvert S\rvert=k-1}I(A_S;B)&\leq \frac{k-1}{k}I(A;B).
\end{align*}
Therefore, there exists a subset $S$ satisfying $I(A_S; B)\leq \frac{k-1}{k}I(A;B)$. 

To verify the optimality of $(k-1)/k$, consider the case of $B=A$. In this case we have $I(A;B)=k$ and $I(A_0;B)=k-1$ for any $2$-to-$1$ function $f$. As another example we can also consider the erasure channel of Example~\ref{exmple1}. In this case, $I(A;B)=k(1-\epsilon)$ and $I(A_0;B)=(k-1)(1-\epsilon)$ for any such~$f$.

\end{proof}

The ratio $(k-1)/k$ in the above theorem, is not desirable since it is close to $1$ for large values of $k$. Furthermore, if we repeatedly use the above theorem to reduce the message-length from $k$ to $k-\ell$, the product
$\prod_{i=k-\ell+1}^k (i-1)/i$ equals $(k-\ell)/k$, which is linear in $\ell$. As a result, if we convert the bound on mutual information to a bound on the total variation distance between $p_{A_0B}$ and $p_{A_0}\times p_B$  (by expressing mutual information in terms of the Kullback--Leibler divergence  and applying Pinsker's inequality), we do not get an exponential decrease of the total variation distance in terms of $\ell$. 
This comparison illustrates the advantage of utilizing the proposed new measure of correlation $V_\alpha$ for privacy amplification.

%***********************************************************************************
\section{Bounding the random coding exponent}\label{sec:exponent}

Decoupling type theorems are widely used in quantum information theory for proving achievability results, e.g., in state merging, the mother protocol, and channel coding, see~\cite{DBWR14} and reference therein. Since Theorem~\ref{thm:main-0} works for all $1\leq \alpha\leq 2$ and not just $\alpha=2$, as in~\cite{Sharma15} we can use our decoupling theorems not only for proving achievability type results but also for proving bounds on the \emph{error exponents}. In the following, we illustrate this application via the problem of \emph{entanglement generation} over a noisy quantum channel and refer to~\cite{Sharma15} for other such examples.

While decoupling is a quantum phenomenon, decoupling-type theorems have also been proven useful in classical information theory. The OSRB method of~\cite{yassaee2014achievability} provides some techniques for proving achievability type results based on decoupling.  Thus our decoupling theorems can be used to prove achievability results in classical network information theory as well. Moreover, as discussed above, we can state effective bounds on the error exponents of such achievability results. In the following, we take this path for the problem of secure communication over wiretap channels and establish an interesting connection between the secrecy exponent for this problem and R\'enyi mutual information according to Csisz\'ar's proposal.

%*****************************************
\subsection{Entanglement generation}
Entanglement generation via a noisy quantum channel is the problem of generating a maximally entangled state of the highest possible dimension between two parties Alice and Bob who are connected by a noisy quantum channel $\cN_{A\rightarrow B}$ from Alice to Bob. To this end, Alice prepares a bipartite state $\rho_{RA}$ send the subsystems $A$ via the channel to Bob. Thus Bob receives the subsystem $B$ of $\mathcal I_R\otimes \cN(\rho_{RA})$. He then applies a decoding map $\mathcal D_{B\to R'}$ and prepares $\mathcal I_R\otimes (\mathcal D\circ \cN)(\rho_{RA})$. The goal of the protocol is that the latter state to be close to a maximally entangled state. A $(\log m, \epsilon)$-code for this problem, with rate $\log m$ and error $\epsilon$, is a choice of the starting state $\rho_{RA}$ and the decoding map $\mathcal D_{B\to R'}$ such that 
$$F\big(  \Phi^m_{RR'},  \mathcal I_R\otimes (\mathcal D\circ \cN)(\rho_{RA})  \big)\geq 1-\epsilon,$$
where $\Phi^m_{RR'}$ is a \emph{maximally entangled state} of local dimension $m$ and $F$ denotes the fidelity function given by $F(\sigma, \tau) = \|\sqrt \sigma\cdot \sqrt \tau\|_1$.
It is well-known that the entanglement generation problem is closely related to quantum commutation over the channel $\cN_{A\to B}$. More precisely, the asymptotic rate of entanglement generation with asymptotically vanishing error equals the capacity of $\cN_{A\to B}$, for which \emph{maximum coherent information} is a lower bound, see e.g.,~\cite{BDL16}.   

\begin{theorem}\label{thm:ent-gen}
Let $\cN_{A\to B}$ be an arbitrary quantum channel. 
Then for any bipartite \emph{pure} state $\ket \psi_{RA}$ and $\alpha\in (1, 2]$ there exists an entanglement generation $(\log m, \epsilon)$ code over $\cN$ if
\begin{align}\label{eq:ent-gen-as}
H_{\alpha}(R|E) - \alpha' \log(1/\epsilon) + 3-\alpha'\geq \log m,
\end{align}
%$$2^{\frac 2\alpha-1} m^{\frac{1}{\alpha'}} W_\alpha(R|E)_\rho \leq  \epsilon,$$
where $\rho_{RE} = \mathcal I_R\otimes \cN^c(\ket \psi\bra\psi_{RA})$ and $\cN^c_{A\to E}$ is the \emph{complementary channel} to $\cN$.
\end{theorem}

Before getting to the proof of this theorem (that is quite standard) let us first state the asymptotic version of the above one-shot bound.

\begin{corollary}\label{cor:ent-gen-asymp}
Let $\cN_{A\to B}$ be an arbitrary quantum channel. 
Then for any bipartite \emph{pure} state $\ket \psi_{RA}$ and $\alpha\in (1, 2]$ there exists an entanglement generation code over $\cN$ with rate $r$ and error rate at most
$$2^{-\frac{n}{\alpha'} \big(   H_\alpha(R|E)_\rho - r  +o(n)\big) },$$
where $\rho_{RE} = \mathcal I_R\otimes \cN^c(\ket \psi\bra\psi_{RA})$ and $\cN^c_{A\to E}$ is the \emph{complementary channel} to $\cN$.

\end{corollary}

\begin{proof}[Proof of Theorem~\ref{thm:ent-gen}]
Since $\ket{\psi}_{RA}$ is a pure state, we may assume without no of generality that $\dim \cH_R=\dim \cH_A=d$. Let $\{\ket 1, \dots, \ket d\}$ be an orthonormal basis for $\cH_R$, and let $\cH_{R'}$ be isomorphic to $\cH_R$. Let 
$$\ket{\Phi^m}_{RR'} = \frac{1}{m}\sum_{i=1}^m \ket i_R\otimes \ket i_{R'},$$
be a maximally entangled state of local dimension $m$, and $\Phi^m_{RR'}=\ket{\Phi^m}\bra{\Phi^m}_{RR'}$ be its associated density matrix.  Let $P_R$ be the following rank $m$ projection
$$P_R= \sum_{i=1}^m \ket i\bra i_R.$$
Let $W_{\cN}: \cH_A\to \cH_B\otimes \cH_E$ be the \emph{Stinespring isometry} associated to $\cN$ so that $\cN(X) = \tr_E\big(  W X W^\dagger   \big)$. Then the \emph{complementary channel} $\cN^c_{A\to E}$ is given by $\cN^c(X) = \tr_B\big(  W_{\mathcal N} X W_{\mathcal N}^\dagger   \big)$.

Let 
$$\ket{\rho}_{RBE} = (I_R\otimes W_{\mathcal N})\ket{\psi}_{RA},$$
and $\rho_{RBE} = \ket \rho\bra \rho_{RBE}$ be its associated density matrix. 
By Corollary~\ref{cor:decoupling-Berta} for every $\alpha\in (1, 2]$ there exists a unitary $U_R$ such that 
\begin{align}\label{eq:dec-B-delta}
\Big\|  \frac{d}{m} (P_RU_R\otimes I_E) \rho_{RE} (U_R^\dagger P_R\otimes I_E) - \frac{1}{m}P_R\otimes \rho_E       \Big\|_1 \leq 2^{\frac 2\alpha-1} m^{\frac{1}{\alpha'}} W_\alpha(R|E)_\rho.
\end{align}
Now define 
$$\ket{\xi'}_{RA}= \sqrt{\frac{d}{m}}  (P_RU_R\otimes I_A)\ket{\psi}_{RA},$$
and let $\ket{\xi} = \frac{1}{\theta} \ket{\xi'}$ where $\theta= \|\ket{\xi'}\|$ is a normalization factor. Also let $\xi_{RA}=\ket \xi\bra \xi_{RA}$ be the corresponding density matrix.  Observe that
\begin{align*}
\mathcal I_R\otimes \cN^c(\xi_{RA}) &= \frac{d}{\theta^2\, m} \tr_{B}\Big( (P_RU_R\otimes W_{\cN})\ket \psi\bra\psi_{RA}        (U_R^\dagger P_R\otimes W^\dagger_{\cN})       \Big)\\
& = \frac{d}{\theta^2\, m}  (P_RU_R\otimes I_A) \rho_{RE} (U^\dagger_RP_R\otimes I_A).
\end{align*}
Then using the Fuchs-van de Graaf inequality and letting $\delta$ to be the right hand side of~\eqref{eq:dec-B-delta} we obtain
\begin{align*}
F\big(\mathcal I_R\otimes \cN^c(\xi_{RA}), \, \frac{1}{m}P_R\otimes \rho_E\big) &\geq 1- \frac{1}{2} \Big\|     \mathcal I_R\otimes \cN^c(\xi_{RA})- \frac{1}{m}P_R\otimes \rho_E       \Big\|_1\\
& = 1-\frac 12  \Big\|    \frac{d}{\theta^2\, m}  (P_RU_R\otimes I_A) \rho_{RE} (U^\dagger_RP_R\otimes I_A)- \frac{1}{m}P_R\otimes \rho_E       \Big\|_1\\
&\geq 1- \Big\|    \frac{d}{ m}  (P_RU_R\otimes I_A) \rho_{RE} (U^\dagger_RP_R\otimes I_A)- \frac{1}{m}P_R\otimes \rho_E       \Big\|_1\\
&\geq 1-\delta,
\end{align*}
where the third line follows from the fact that for any two density matrices $\sigma, \sigma'$  and $c\in \mathbb R$ we have $\|\sigma-\sigma'\|_1\leq 2\|c\sigma-\sigma'\|_1$ whose proof can be found in~\cite{HHWY08}. 

Observe that $I_R\otimes W_{\cN} \ket{\xi}_{RA}$ is a purification of $\mathcal I_R\otimes\cN^c(\xi_{RA})$ and $\ket{\Phi^m}_{RR'}$ is a purification of $\frac1mP_R$. Fix some purification $\ket\tau_{EE'}$ of $\rho_E$. Then by Uhlmann's theorem there exists an isometry $Z: \cH_B\to \cH_{R'}\otimes \cH_{E'}$ such that 
$$F\big(\mathcal I_R\otimes \cN^c(\xi_{RA}), \, \frac{1}{m}P_R\otimes \rho_E\big) = \big| \bra{\Phi^{m}}_{RR'}\otimes \bra{\tau}_{EE'}\,  ( I_R\otimes ZW_{\cN} ) \,\ket \xi_{RA}   \big|,$$
and then by the monotonicity of fidelity
\begin{align*}
1-\delta & \leq  F\big( \ket{\Phi^m}_{RR'}\otimes \ket{\tau}_{EE'},\,  ( I_R\otimes ZW_{\cN} ) \,\ket \xi_{RA} \big)\\
& \leq F\big(\Phi^m_{RR'}, \mathcal I_R\otimes (\mathcal D\circ \mathcal N ) (\xi_{RA})    \big), 
\end{align*}
where $\mathcal D: \fL(B)\to \fL(R')$ is given by $\mathcal D(X) = \tr_{E'}(ZXZ^\dagger)$. Thus the only remaining step is to show that $\epsilon\geq \delta$. That is, we need to verify that 
$$2^{\frac 2\alpha-1} m^{\frac{1}{\alpha'}} W_\alpha(R|E)_\rho\leq \epsilon.$$
Using Proposition~\ref{prop2} and the fact that $H_{\alpha}(R|E)_\rho\leq \log d$, the above inequality is implied once we have
$$2^{ -\frac{1}{\alpha'}  \big( H_{\alpha}(R|E)_\rho + 1 -\log m - \alpha'(2/\alpha-1)     \big)}\leq \epsilon,$$
which is equivalent to our assumption~\eqref{eq:ent-gen-as}. 

\end{proof}

%*****************************************
\subsection{Statistics of random binning }\label{sec:binning}
Decoupling-type theorems are also utilized in classical information theory for proving achievability results via the method of OSRB~\cite{yassaee2014achievability}. Moreover, as in the quantum case for the problem of entanglement generation, our decoupling theorems can be used for proving bounds on the error exponents in such achievability results. %In the rest of this section we follow this path for the problem of information transmission over the wiretap channel, 
Yet in the classical case we are able to prove even stronger error exponents, comparing to that of Corollary~\ref{cor:ent-gen-asymp}, by replacing R\'enyi information measures according to the proposal of Sibson, by those of Csisz\'ar. Thus here we prove an asymptotic version of our decoupling theorem in the classical case in which surprisingly Csisz\'ar's proposal of $\alpha$-R\'enyi mutual information appears. Next, we will apply this result to the problem of the capacity of the wiretap channel.

Let $(A^n, B^n)$ be i.i.d.\ classical random variables distributed according to $p_{AB}$:
$$p(a^nb^n)=\prod_{i=1}^n p(a_i b_i).$$
Suppose that we randomly (and uniformly) bin the set $\mathcal{A}^n$ into $2^{nR}$ bins and let $A_0$ to denote the bin index. Finding the correlation between the bin index $A_0$ and $B^n$ (averaged over all random bin mappings) is of interest, see \cite{yassaee2014achievability}. It is known that  if the binning rate $R$ is below the Slepian-Wolf rate, i.e., $R<H(A|B)$, the average total variation distance $\|p_{A_0B^n}-p_{A_0}\times p_{B^n}\|_1=V_1(A_0;B^n)$  vanishes asymptotically as $n$ tends to infinity.

Here we are interested in the same question as above when we replace $V_1(A_0;B^n)$ with the correlation measure $V_\alpha(A_0;B^n)$ for some $\alpha\in (1,2]$. Our tool for answering this question is Theorem~\ref{thm:main-3}, yet this theorem is applicable only if the first variable is distributed uniformly. For this reason, we do not assume that $A^n$ is i.i.d., but is completely uniform on a type set. 

Let $p_{AB}$ be a bipartite distribution such that $p_A(a)$ is a rational number for all $a\in\mathcal{A}$. In the following, let $n$ be some natural number such that $np(a)$ is an integer for all $a\in \mathcal A$. For such $n$, let $\mathcal{T}_{n}(p_A)\subseteq \mathcal A^n$ be the set of all sequences $a^n$ of length $n$ whose empirical distribution (type) is equal to $p_A$, \emph{i.e.,} each symbol $a'\in\mathcal{A}$ occurs  exactly $np(a')$ times in sequence $a^n$. Instead of the i.i.d.\ distribution on $A^n$, let $A^n$ be uniformly distributed over $\mathcal{T}_{n}(p_A)$. The conditional distribution of $B^n$ given $A^n$ is still assumed to be
$$p(b^n|a^n)=\prod_{i=1}^n p(b_i|a_i).$$
For random binning, we use a randomly chosen $k$-to-$1$ function $f$ on $\mathcal T_n(p_A)\subseteq \mathcal A^n$ and let $A_0=f(A^n)$. We call this a \emph{regular random binning}. This corresponds to a binning procedure with rate 
\begin{align}
R=\frac 1n\log\Big(\frac{|\mathcal{T}_{n}(p_A)|}{k}\Big).
\label{eqnRr}\end{align}

\begin{theorem} \label{thm2ere34} 
Let $A^n$ be uniformly distributed over $\mathcal{T}_{n}(p_A)$ and 
$$p_{B^n|A^n}=\prod_{i=1}^n p_{B_i|A_i}.$$
Also let $k$ be an integer that divides $|\mathcal T_n(p_A)|$ and define $R$ by~\eqref{eqnRr}. Then for every $\alpha\in(1,2]$ we have 
\begin{align}
\mathbb{E}\big[V_\alpha(A_0;B^n)\big]&\leq 
%2^{-\frac {n}{\alpha'}\left(R_0-R\right)+o(n)}=
2^{-\frac{n}{\alpha'}\big(H(A)-I^c_{\alpha}(A;B)-R+o(n)\big)},\label{eqne3334}
\end{align}
where $A_0=f(A^n)$, the average is taken over all $k$-to-$1$ functions $f:\mathcal T_n(p_A)\to \mathcal A_0$ (i.e., over all regular random bin mappings~$f$) and  $I_{\alpha}^{\mathrm{c}}(A;B)$ is the $\alpha$-R{\'e}nyi mutual information according to Csisz\'ar's proposal~\cite[Eq. 29]{verdu2015alpha} defined by 
$$I_{\alpha}^{\mathrm{c}}(A;B)=\min_{q_B}\sum_{a}p(a)D_\alpha\left(p_{B|a}\parallel q_B\right).$$
In particular, the average correlation $\mathbb{E}[V_\alpha(A_0;B^n)]$
 vanishes as $n$ tends to infinity if 
\begin{align*}
R&<H(A)-I_{\alpha}^{\mathrm{c}}(A;B).
\end{align*} 
Furthermore, we have
\begin{align}
\mathbb{E}\big[\big\|p_{A_0B^n}-p_{A_0}\times p_{B^n}\big\|_1\big]&\leq 
2^{-\max_{1\leq\alpha\leq 2}\left\{
\frac{n}{\alpha'}\big(H(A)-I^c_{\alpha}(A;B)-R+o(n)\big)\right\}}\nonumber
\\&=2^{-n\big(\min_{q_{AB}: q_A=p_A}D(q_{B|A}\|p_{B|A}|p_A)+[\frac12H(A|B)_{q}-R]_{+}+o(n)\big)}. \label{eqn:slkdfj4225}
\end{align}

\end{theorem}

From \cite[Eq. 24]{csiszar1995generalized}, we have $H(A)\geq I_{\alpha}^{\mathrm{c}}(A;B)$ with equality when $B=A$. Thus, the above bound $H(A)-I_{\alpha}^{\mathrm{c}}(A;B)$ on the binning rate is always non-negative.
Moreover, since $I_{\alpha}^{\mathrm{c}}(A;B)\geq I(A;B)$, we have $H(A)-I_{\alpha}^{\mathrm{c}}(A;B)\leq H(A)-I(A;B)=H(A|B)$. 
Hence, the bound given in the statement of the theorem on $R$ does not exceed $H(A|B)$, the conditional Slepian-Wolf rate, as expected.

\begin{remark}  
To the best of our knowledge, the generalized cut-off rates of Csisz\'ar for the dependencies of random bin indices are not defined or studied in the literature. However, we point out that resolvability exponents are studied in \cite{parizi2017exact, hayashi2006general, endo2014reliability, yagli2018exact}.
In particular, \cite{yagli2018exact} finds the following resolvability exponent for i.i.d.~codewords:
	\begin{align}
		\alpha(R',P_X,P_{Y|X})&= \max_{\lambda\in [0,1] } \left\{ \frac\lambda2 R' - \log \bbE\left[ \left( \bbE \left[\exp\left(\frac{\lambda}{2-\lambda}\, \imath_{X;Y}(X;Y) \right ) \big| Y\right] \right)^{\frac{2-\lambda}{2}} \right]  \right\} \text{.} \label{eqn:def:alpha(R,P_X,P_{Y|X})_the dual}
	\end{align}
With the change of variable ${1}/{\alpha'}=\lambda/2$, the above expression equals
$$\max_{\alpha\in[1,2]}
\frac{1}{\alpha'}\left(R'-I^s_{\alpha}(A;B)\right),$$
where $I^s_{\alpha}(A;B)$ is the $\alpha$-R{\'e}nyi mutual information according to Sibson's proposal. To relate the resolvability problem and our problem, let $R'=H(A)-R$. Then, we see that the exponent of \cite{yagli2018exact} has the same form as our exponent, except that our $\alpha$-R{\'e}nyi mutual information is computed according to Csiszar's proposal which result in stronger bounds. 
%This implies that our exponents are better when we apply the results to the exponents of the wiretap channel.
\end{remark}

\begin{proof}[Proof of Theorem \ref{thm2ere34}] From Theorem~\ref{thm:main-3}, with a randomly chosen $k$-to-$1$ function $f$ acting on $\mathcal{T}_{n}(p_A)$, 
we have
\begin{align}
\mathbb{E}[V_\alpha(A_0;B^n)]&\leq 
2^{\frac 2\alpha-1} k^{-\frac{1}{\alpha'}}V_\alpha(A^n;B^n)
\leq 2^{\frac 2\alpha-1} k^{-\frac{1}{\alpha'}}\Big(2^{\frac {1}{\alpha'}I_\alpha(A^n;B^n)}+1\Big),\label{eqn:ineqaulity3234}
\end{align}
where for the second inequality we use Propositin~\ref{prop2}.

%We have
%\begin{align}V_\alpha(A^n;B^n)\leq 2^{\frac {1}{\alpha'}I_\alpha(A^n;B^n)}+1.\label{eqn:ineqaulity34}
%\end{align}
Note that the distribution of $(A^n, B^n)$ is not i.i.d., so $I_\alpha(A^n;B^n)$ is not equal to $nI_\alpha(A; B)$. It is shown in Lemma~\ref{lemma:reneeo} below that 
\begin{align}\label{eq:lem-28}
2^{\frac {1}{\alpha'}I_\alpha(A^n;B^n)}=2^{\frac {n}{\alpha'}\big(I_{\alpha}^{\mathrm{c}}(A;B)+o(n)\big)}.
\end{align}
%Since $\frac {1}{\alpha'}I_\alpha(A^n;B^n)$ tends to infinity as $n$ tends to infinity, from \eqref{eqn:ineqaulity34} we have
%$$V_\alpha(A^n;B^n)\leq 2^{\frac {n}{\alpha'}I_{\alpha}^{\mathrm{c}}(A;B)+o(n)}.$$
Then, from \eqref{eqn:ineqaulity3234} we have
\begin{align}
\mathbb{E}[V_\alpha(A_0;B^n)]&\leq k^{-\frac{1}{\alpha'}}2^{\frac {n}{\alpha'}\big(I_{\alpha}^{\mathrm{c}}(A;B)+o(n)\big)}\nonumber\\
&=2^{-\frac {n}{\alpha'} \big(  \frac{1}{n}\log|\mathcal T_n(p_A)| -I_{\alpha}^{\mathrm{c}}(A;B) -R+o(n)\big)}
\nonumber\\
&=2^{-\frac {n}{\alpha'} \big(  H(A) -I_{\alpha}^{\mathrm{c}}(A;B) -R+o(n)\big)}.
\label{eqne34}
\end{align}

To prove equation~\eqref{eqn:slkdfj4225}, applying Proposition~\ref{prop:non-decreasing}, 
%for any $\alpha\in(1,2]$ we have
%\begin{align}
%\mathbb{E}_f\|p_{A_0B^n}-p_{A_0}p_{B^n}\|\leq 
%2^{-\frac {n}{\alpha'}\left(R_0-R\right)+o(n)}
%=2^{-\frac{n}{\alpha'}\left(H(A)-I^c_{\alpha}(A;B)-R+o(n)\right)}
%. \label{eqn:slkdfj45}
%\end{align}
it suffices to verify that 
\begin{align}
\max_{1\leq\alpha\leq 2}\frac{1}{\alpha'}\big(H(A)-I^c_{\alpha}(A;B)-R\big)=\min_{q_{AB}: q_A=p_A}D\big(q_{B|A}\|p_{B|A}\,|\,p_A\big)+\big[\frac12H(A|B)_{q}-R\big]_{+}.\label{eq:crmi-equiv}
\end{align} 
To see this, we use~\cite[Eq. 7]{6034266}
\begin{align}\label{eq:CMI-091}
I_{\alpha}^{\mathrm{c}}(A;B)=\max_{q_{AB}: q_A=p_A}\Big(
I(A;B)_q-\alpha'D\big(q_{B|A}\|p_{B|A}\, |\, p_A\big)\Big).
\end{align}
Therefore,
\begin{align*}&\max_{1\leq\alpha\leq 2}\frac{1}{\alpha'}\big(H(A)-I^c_{\alpha}(A;B)-R\big)
\\&=
\max_{1\leq\alpha\leq 2}\,\min_{q_{AB}: q_A=p_A}
\frac{1}{\alpha'}
\Big(H(A)_p-I(A;B)_q+\alpha'D\big(q_{B|A}\|p_{B|A}\, |\, p_A\big)-R\Big)
\\&=
\max_{1\leq\alpha\leq 2} \,\min_{q_{AB}: q_A=p_A}
\frac{1}{\alpha'}
\Big(H(A|B)_q+\alpha'D\big(q_{B|A}\|p_{B|A}\, |\, p_A\big)-R\Big)
\\&=\max_{0\leq\zeta\leq \frac 12}\,\min_{q_{AB}: q_A=p_A}
\zeta
\big(H(A|B)_q-R\big)+D\big(q_{B|A}\|p_{B|A}\, |\, p_A\big).
\end{align*}
Then~\eqref{eq:crmi-equiv} follows once we  exchange the maximum and minimum in the above equation. This exchange is possible since the expression is easily seen to be convex in $q_{B|A}$ and linear in $\zeta$ since for $q_{AB}= p_A\times q_{B|A}$ we have
$$\zeta
H(A|B)_q+D\big(q_{B|A}\|p_{B|A}\, |\, p_A\big) = \xi H(A)_p - (1-\zeta)H(B|A)_q - \zeta H(B)_q - \sum_{a, b} p(a)q(b|a) \log p(b|a).$$
\end{proof}

It remains to verify~\eqref{eq:lem-28} to complete the above proof. 

\begin{lemma}\label{lemma:reneeo} Let $p_{AB}$ be an arbitrary joint probability distribution. Let $A^n$ be uniform over $\mathcal{T}_{n}(p_A)$ and 
$$p(b^n|a^n)=\prod_{i=1}^n p(b_i|a_i).$$
Then, for any $\alpha>1$ we have
$$\lim_{n\rightarrow\infty}\frac{1}{n}I_\alpha(A^n;B^n)=I_{\alpha}^{\mathrm{c}}(A;B).$$
\end{lemma}

\begin{proof} 
We use  standard arguments from the method of types. For simplicity of notation, 
for two sequences $\{x_n:\, n\geq 1\}$ and $\{y_n:\, n\geq 1\}$, we use $x_n\circeq y_n$ to denote
$$\lim_{n\rightarrow \infty} \frac{1}{n}  x_n = \lim_{n\rightarrow \infty} \frac{1}{n}  y_n.$$
Then we have $\log |\mathcal{T}_{n}(p_A)|\circeq nH(A)_p$. 
Using~\eqref{eq:RMI-sibson} we have
\begin{align}
\frac1{\alpha'} I_\alpha(A^n;B^n)&=\log\Bigg(\sum_{b^n} \bigg(\sum_{a^n\in\mathcal{T}_{n}(p_A)} 
\frac{1}{{|\mathcal{T}_{n}(p_A)|}} p(b^n|a^n)^{\alpha}\bigg)^{1/\alpha}\Bigg)\nonumber
\\&=-\frac{1}{\alpha}
\log |\mathcal{T}_{n}(p_A)|
+\log
\Bigg(\sum_{b^n} \bigg(\sum_{a^n\in\mathcal{T}_{n}(p_A)} p(b^n|a^n)^{\alpha}\bigg)^{1/\alpha}\Bigg)\nonumber
\\&\circeq -\frac{n}{\alpha}H(A)_p+ 
\log\Bigg(\sum_{b^n} \bigg(\sum_{a^n\in\mathcal{T}_{n}(p_A)} p(b^n|a^n)^{\alpha}\bigg)^{1/\alpha}\Bigg).
\label{eql42}
\end{align}
%We note that $p(b^n|a^n)$ depends only on the joint type of $(a^n, b^n)$. Fixing this joint type $q_{AB}$, whose marginal is $q_A=p_A$, for any $(a^n, b^n) \in \mathcal T^{\mathcal{A\times B}}_n(q_{AB})$ we have
%$$p(b^n| a^n) = 2^{\sum_{a, b}    nq(a, b) \log p(b|a)  }.$$
%Therefore, using the fact that the number of such types $q_{AB}$ is polynomial in $n$ we have
%\begin{align*}
%2^{\frac1{\alpha'} I_\alpha(A^n;B^n)}& \circeq 2^{-\frac{n}{\alpha}H(A)_p} \sum_{q_{AB}:\, q_A=p_A} \big|\mathcal T^{\mathcal B}_n(q_B)\big|
%\end{align*}
Observe that for any $b^n\in \mathcal B^n$, the expression $\sum_{a^n\in\mathcal{T}_{n}(p_A)} p(b^n|a^n)^{\alpha}$ depends only on the type of $b^n$ (since $\mathcal T_n(p_A)$ is permutation invariant). Thus letting $b_0^n \in \mathcal T_n(q_B)$ to be of type $q_B$ we define
\begin{align}\label{eq:F-q-B}
F(q_B)=\sum_{a^n\in\mathcal{T}_{n}(p_A)} p(b_0^n|a^n)^{\alpha}.
\end{align}
Then denoting the set of all types in $\mathcal B^n$ by $\Upsilon_n(\mathcal B)$, the second term on the right hand side of~\eqref{eql42} can be expressed as
\begin{align}\nonumber
\log\Bigg(\sum_{b^n} \bigg(\sum_{a^n\in\mathcal{T}_{n}(p_A)} p(b^n|a^n)^{\alpha}\bigg)^{1/\alpha}\Bigg)
&=\log\Bigg(\sum_{q_B\in \Upsilon_n(\mathcal B) }|\mathcal{T}_{n}(q_B)|\cdot F(q_B)^{1/\alpha} \Bigg)\\
&\circeq \max_{q_B\in \Upsilon_n(\mathcal B)} \log\Big(\big|\mathcal{T}_{n}(q_B)\big|\cdot F(q_B)^{1/\alpha}\Big)\nonumber
\\&\circeq \max_{q_B\in \Upsilon_n(\mathcal B)} nH(B)_q + \frac{1}{\alpha}\log F(q_B),\nonumber
\end{align}
where in the second line we use the fact that there are polynomially many types in $\Upsilon_n(\mathcal B)$. 

The next step is to compute $F(q_B)$. Since~\eqref{eq:F-q-B} depends only on the type of $b_0^n\in \mathcal T_n(q_B)$ we have
\begin{align}\nonumber
F(q_B)&=\sum_{a^n\in\mathcal{T}_{n}(p_A)} p(b_0^n|a^n)^{\alpha}
\\&=\frac{1}{|\mathcal{T}_{n}(q_B)|}\sum_{b^n\in\mathcal{T}_{n}(q_B)} \sum_{a^n\in\mathcal{T}_{n}(p_A)} p(b^n|a^n)^{\alpha}.\nonumber
\end{align}
Let us denote the joint type of $(a^n, b^n)\in \mathcal{T}_{n}(p_A)\times \mathcal{T}_{n}(q_B)$ by $q_{AB}$. Note that the marginal type of $a^n$ is $p_A=q_A$ and $q_{AB}$ is an ``extension" of $q_B$. Denoting the set of all such joint types by $\widetilde \Upsilon_n(q_B)=\Upsilon_n(\mathcal A\times \mathcal B| p_A, q_B)$, for any sequence $(a^n, b^n)$ of joint type $q_{AB}\in \widetilde \Upsilon_n(q_B)$ the value of $p(b^n|a^n)^{\alpha}$ equals $\prod_{a,b}p(b|a)^{n\alpha q(a,b)}$. Therefore, we can  compute $F(q_B)$ by splitting the sum over different joint types. By a similar argument as before, to compute the exponential growth of the summation, we should only consider the ``dominant" type. Therefore, 
\begin{align}
\log F(q_B)&=-\log |\mathcal{T}_{n}(q_B)| + \log \bigg(\sum_{b^n\in\mathcal{T}_{n}(q_B)} \sum_{a^n\in\mathcal{T}_{n}(p_A)} p(b^n|a^n)^{\alpha}\bigg)\nonumber
\\&=
-\log |\mathcal{T}_{n}(q_B)| + \log \bigg(\sum_{q_{AB}\in \widetilde\Upsilon_n(q_B)}
\big|\mathcal{T}_{n}(q_{AB})\big|\cdot \prod_{a,b}p(b|a)^{n\alpha q(a,b)}\bigg)\nonumber
\\&\circeq
-nH(B)_q+ \max_{q_{AB}\in \widetilde\Upsilon_n(q_B)}
\log\Big(\big|\mathcal{T}_{n}(q_{AB})\big|\cdot \prod_{a,b}p(b|a)^{n\alpha q(a,b)}\Big)\nonumber
\\&\circeq
-nH(B)_q+ \max_{q_{AB}\in \widetilde\Upsilon_n(q_B)}
nH(AB)_q + n\alpha \sum_{a,b}q(ab)\log p(b|a).\nonumber
\end{align}
%With a change of notation we have
%\begin{align}F(q_B)&\circeq \max_{q(a,b)=p(a)q(b|a)}
%2^{-nH_q(B)}2^{nH_q(A,B)}2^{\sum_{a,b}n\alpha q(a,b)\log p(b|a)}\nonumber
%\\&=\max_{q(a,b)=p(a)q(b|a)}
%2^{nH_q(A|B)}2^{\sum_{a,b}n\alpha q(a,b)\log p(b|a)}\nonumber
%\end{align}
Putting everything together, we have
\begin{align*}
\frac1{\alpha'} I_\alpha(A^n;B^n)&
%\circeq
%-\frac{n}{\alpha}H(A)_p + \max_{q_B\in\Upsilon_n(\mathcal B)} \max_{q_{AB}\in \widetilde\Upsilon_n(q_B)} \bigg(nH(B)_q \\
%& \qquad \quad + \frac{1}{\alpha} \Big(
% -nH(B)_q+ nH(AB)_q+n\alpha \sum_{a,b} q(ab)\log p(b|a)\Big)\bigg)\\
 \circeq
-\frac{n}{\alpha}H(A)_p + \max_{q_B\in\Upsilon_n(\mathcal B)} \max_{q_{AB}\in \widetilde\Upsilon_n(q_B)} \bigg(nH(B)_q  + \frac{n}{\alpha}H(A|B)_q+n \sum_{a,b} q(ab)\log p(b|a)\bigg)\\
 & = \max_{q_{AB}: q_A=p_A}\bigg(-\frac{n}{\alpha}H(A)_p + n 
 H(B)_q+ \frac{n}{\alpha}H(A|B)_q+n\sum_{a,b} q(ab)\log p(b|a)\bigg).
\end{align*}
Therefore,
\begin{align}\lim_{n\rightarrow\infty}\frac{1}{n}I_\alpha(A^n;B^n)&=\alpha'
\max_{q_{AB}: q_A=p_{A}}\Bigg(
-\frac{1}{\alpha}H(A)_q+H(B)_q+{\frac{1}{\alpha}H(A|B)_q}+{\sum_{a,b} q(ab)\log p(b|a)}\Bigg)\nonumber\\
%\end{align}
%We now simplify the maximization:
%\begin{align}
%\alpha'\max_{q_{AB}: q_A=p_{A}}&\Bigg(
%-\frac{1}{\alpha}H(A)_q+H(B)_q+{\frac{1}{\alpha}H(A|B)_q}+{\sum_{a,b} q(ab)\log p(b|a)}\Bigg)\nonumber\\
&=\alpha'\max_{q_{AB}: q_A=p_{A}}\Bigg(
\frac{1}{\alpha'}I(A;B)_q+H(B|A)_q+{\sum_{a,b} q(ab)\log p(b|a)}\Bigg)\nonumber\\
&=\alpha'\max_{q_{AB}: q_A=p_{A}}\Bigg(
\frac{1}{\alpha'}I(A; B)_q+{\sum_{a,b} q(ab)\log \frac{p(b|a)}{q(b|a)}}\Bigg)\nonumber
\\&=\max_{q_{AB}: q_A=p_{A}}\Bigg(
I(A;B)_q-\alpha' D\big(q_{B|A}\|p_{B|A}\, |\, p_A\big)\Bigg).\nonumber
\end{align}
The last expression, as mentioned in~\eqref{eq:CMI-091}, equals $I_{\alpha}^{\mathrm{c}}(A;B)$. 
\end{proof}

%*****************************
\subsection{The wiretap channel}

A wiretap channel is determined by a bipartite conditional distribution $p_{YZ|X}$ in which $X$ is the input of the channel, output $Y$ is received by the legitimate receiver and output $Z$ is received by an eavesdropper. The goal of communication over a wiretap channel is to securely send information to the legitimate receiver. It is well-known that for any input distribution $p_X$, the rate $I(X;Y)-I(X;Z)$ is achievable. Our goal here is to establish a bound on the secrecy exponent of random coding over a wiretap channel.  

%Take some type $p_X$. A constant-composition random wiretap code is defined as follows: let $(R_1, R_2)$ be positive numbers satisfying $R_1+R_2<I(X;Y)$ and $R_2>I(X;Z)$. Here $R_1$ is the message rate. We uniformly and randomly choose $2^{n(R_1+R_2)}$ distinct codewords sequences $x^n(m, g)$ for $m\in[2^{nR_1}], g\in[2^{nR_2}]$ from  $\mathcal{T}_{n}(p_X)$, the sequences of type $p_X$. To send a message $m$, the transmitter chooses $G\in[2^{nR_2}]$ uniformly at random and sends the codeword $x^n(m, G)$ over the wiretap channel. We say that a secrecy exponent $\kappa>0$ is achievable with constant-composition random codes if the average of $V_\alpha(M;Z^n)$ over all random codebooks, is at most $2^{-n\kappa+o(n)}$.

\begin{theorem}\label{thm:wiretap}
Let $p_{YZ|X}$ be an arbitrary wiretap channel and take $\alpha\in(1,2]$. Then for any input distribution $p_X$ there exists a code for reliably sending message $M$ of rate $R$ over the channel (with asymptotically vanishing error) such that
\begin{align} \label{exponent343}
V_\alpha(M;Z^n)\leq 2^{-\frac {n}{\alpha'}\big(I(X;Y)-I_{\alpha}^{\mathrm{c}}(X;Z)-R+o(n)\big)}.
\end{align}
In particular, for such a code we have
\begin{align}
\big\|p_{MZ^n}-p_M\times p_{Z^n}\big\|_1\leq 2^{-\frac {n}{\alpha'}\big(I(X;Y)-I_{\alpha}^{\mathrm{c}}(X;Z)-R+o(n)\big)}.
\end{align}
\end{theorem}

\begin{proof}
By a continuity type argument we can assume with no loss of generality that $p(x)$ for any $x\in \mathcal X$ is a rational number, and in the following, we take $n$ to be a sufficiently large number such that $np(x)$ is a natural number for all $x$.  
 Let  $\mathcal{T}_{n}(p_X)\subseteq \mathcal X^n$ be the set of sequences of type $p_X$, and let $X^n$ be uniformly distributed over $\mathcal{T}_{n}(p_X)$.

Choose positive reals $R_1, R_2, R_3$, which may depend on $n$, such that 
\begin{itemize}
\item $R_1=R+o(n)$, 
\item $R_3> H(X|Y)$,  
\item $R_1+R_3< H(X) - I^c_\alpha(X; Z)$
\item $2^{nR_{i}}$ is an integer for $i=1, 2, 3$ and 
$$|\mathcal{T}_{n}(p_X)|=\prod_{i=1}^3 2^{nR_{i}}.$$
\end{itemize}
Observe that if $R< I(X; Y)- I^c_\alpha(X; Z)$ such a triple $(R_1, R_2, R_3)$ exists. 

Let $f=(m, g, u):\mathcal T_n(p_X)\to \big[2^{nR_1}\big] \times \big[2^{nR_2}\big]\times \big[2^{nR_3}\big]$ be a random $1$-to-$1$ function (relabeling), and define $M=m(X^n)$, $G=g(X^n)$, $U=u(X^n)$. Note that, for example, $(m, g):\mathcal T_n(p_X)\to \big[2^{nR_1}\big] \times \big[2^{nR_2}\big]$ is a random $2^{\big[nR_3\big]}$-to-1 function. 
Moreover, since $X^n$ is distributed uniformly over $\mathcal{T}_{n}(p_X)$, random variables $M$, $G$ and $U$ will be uniform and mutually independent. 
%Choosing a randomly and uniformly chosen labeling of $\mathcal{T}_{n}(p_X)$ by a triple $(m, g, u)$ corresponds to a regular random binning function. For instance, $(m,g)$ is a regular random binning of rate $\tilde{R}=R_1+R_2$ of $X^n$ by applying a randomly chosen $\frac{|\mathcal{T}_{n}(p_X)|}{2^{n(R_{1n}+R_{2n})}}$-to-$1$ function on the set $\mathcal{T}_{n}(p_X)$. 

%Let $\tilde{f}$ denote the random labeling of $\mathcal{T}_{n}(p_X)$ by a triple $(m, g, u)$. 
If $R_{3}>H(X|Y)$, having access to $(U, Y^n)$, the legitimate receiver can decode $X^n$ with a vanishing average error probability: 
%That is, letting $P_{\mathsf{error}}$ to be the error probability of decoding $X^n$ given $(U, Y^n)$ we have
\begin{align}
\mathbb{E}\big[\mathrm{Pr}(\mathsf{error})\big]\rightarrow 0,
\label{eqne333322s2334}
\end{align}
as $n$ goes to infinity, where the average is taken over the random choice of $f$. 
Next, by Theorem~\ref{thm2ere34} since $R_{1}+R_{3}<H(X)-I_{\alpha}^{\mathrm{c}}(X;Z)$, we have
\begin{align}
\mathbb{E}\big[V_\alpha(M,U~;Z^n)\big]&\leq 
2^{-\frac {n}{\alpha'}\big(H(X)-I_{\alpha}^{\mathrm{c}}(X;Z)-R_1-R_3+o(n)\big)}.
\end{align}
On the other hand, by Theorem~\ref{thm:conditional-V-1} we obtain 
\begin{align}\nonumber
\mathbb{E}\big[V_\alpha(M; Z^n|U)\big]&\leq 2^{\frac2\alpha -1}\mathbb{E}\big[V_\alpha(M,U~;Z^n)\big]
\\&\leq  
2^{-\frac {n}{\alpha'}\big(H(X)-I_{\alpha}^{\mathrm{c}}(X;Z)-R_1-R_3+o(n)\big)}.\label{eqne33das32334}
\end{align}
Therefore, using~\eqref{eqne333322s2334} and~\eqref{eqne33das32334}, and Markov's inequality together with a union bound, for any $\epsilon>0$ and sufficiently large $n$, there exists $u\in \big[2^{nR_3}\big]$ and a random labeling $f_0$ such that  
\begin{align}
\mathrm{Pr}(\mathsf{error}| f_0, U=u)\leq\epsilon,
\label{eqnsds1f}
\end{align}
and
\begin{align}
V_\alpha(M; Z^n|f_0, U=u)&\leq 
2^{-\frac {n}{\alpha'}\big(H(X)-I_{\alpha}^{\mathrm{c}}(X;Z)-R_1-R_3+o(n)\big)}.\label{eqnsds2f}
\end{align}
%To see this, note that from Markov's inequality, the probability that a randomly chosen $\tilde{f}$ and $U$ satisfies 
%\begin{align}P_{\mathsf{error}|\tilde{f}, U}\leq\epsilon\end{align}
%is greater than $2/3$ for large enough $n$. Similarly,  the probability that a randomly chosen $\tilde{f}$ and $U$ satisfies 
%\begin{align}
%V_\alpha(M;Z^n|\tilde{f}, U)&\leq 
%2^{-\frac {n}{\alpha'}\left(H(X)-I_{\alpha}^{\mathrm{c}}(X;Z)-R_1-R_2\right)+o(n)}.
%\end{align}
%is greater than $2/3$ for large enough $n$. Therefore, one can find some  $U=u$ and a random labeling $\tilde{f}^*$ such that both inequalities \eqref{eqnsds1f} and \eqref{eqnsds2f} are satisfied.
Now, as in \cite{yassaee2014achievability}, the code can be constructed as follows. We treat  $M$ as the message (which is distributed uniformly), select $G$ uniformly at random and independent of $M$ and transmit the codeword $X^n=f_0^{-1}(M, G, u)$. The legitimate receiver can decode $M$ with an asymptotically vanishing error because of~\eqref{eqnsds1f}, and the eavesdropper would gain no information about $M$ due to~\eqref{eqnsds2f}.
\end{proof}

%%%%%%%%%%%%%%%%%%%%%%
%Appendix

\appendix

\vspace{.4in}
\begin{center}
\textbf{\LARGE Appendix}
\end{center}

\section{Riesz-Thorin interpolation theorem}\label{app:interpolation}
In this appendix, we provide a very brief simplified overview of the theory of interpolation spaces and the Riesz-Thorin theorem. For a detailed introduction to the subject, we refer to~\cite{Lunardi}.

Let $X$ be a finite dimensional complex vector space which can be equipped with different norms. 
Let us denote this vector space with two different such norms on it by $X_0, X_1$. Thus $X_0$ and $X_1$ are Banach spaces.   Then the theory of complex interpolation provides us with a method for constructing \emph{intermediate} Banach spaces $X_\theta$ for all $\theta\in [0,1]$. A typical example for such an interpolation family is the $\ell_p$ spaces. For $i=0,1$, letting $X_i=\ell_{p_i}(A)$ be the vector space $X=\ell(A)$ equipped with the $p_i$-norm, then the interpolating space $X_\theta$, for $\theta\in [0,1]$, is equal to $\ell_{p_\theta}(A)$  where $p_\theta$ is given by
\begin{align}\label{eq:def-p-theta}
\frac{1}{p_\theta} = \frac{1-\theta}{p_0} + \frac{\theta}{p_1}.
\end{align}
Similarly, the non-commutative spaces $L_{p_\theta}(A)$ form an interpolation family for $\theta\in [0,1]$ if $p_\theta$'s satisfy the above equation. A more sophisticated example is the family of vector-valued spaces;  For example, the interpolation of the $(p_0, q_0)$-norm and the $(p_1, q_1)$-norm is the $(p_\theta, q_\theta)$-norm where both $p_\theta$ and $q_\theta$ satisfy~\eqref{eq:def-p-theta}.

%Let $X, Y$ be two Banach spaces and let $T:X\to Y$ be a linear map. Then the operator norm of $T$ is defined by
%$$\|T\|_{X\to Y} = \sup_{x\in X} \frac{\|T(x)\|}{\|x\|}.$$
%Here, of course, $\|x\|$ is defined with respect to the norm of $X$, and $T(x)$ is defined with respect to the norm of $Y$. 
We can now state a version of the Riesz-Thorin interpolation theorem.

\begin{theorem}[Riesz-Thorin theorem]  \label{thm:Riesz-Thorin}
Let $\{X_\theta:\, \theta\in [0,1]\}$ and $\{Y_\theta:\, \theta\in [0,1]\}$ be two families of interpolation spaces over finite dimensional vector spaces $X$ and $Y$ respectively. Assume that $T: X\to Y$ is a linear map which can be regarded as a continuous map from $X_\theta$ to $Y_\theta$. Then for any $\theta\in [0,1]$ we have
$$\|T\|_{X_\theta \to Y_\theta} \leq \|T\|_{X_0\to Y_0}^{1-\theta} \cdot \|T\|_{X_1\to Y_1}^{\theta},$$
where $\|T\|_{X_\theta\to Y_\theta}$ is the operator norm given by
$$\|T\|_{X_\theta\to Y_\theta}= \sup_{x\in X} \frac{\|T(x)\|_{Y_{\theta}}}{\|x\|_{X_\theta}}.$$
\end{theorem}

Restricting to the example of vector-valued $p$-norms, we obtain the following.

\begin{corollary}
Let $\Xi:\fL(BA)\to \fL(CBA)$ be a linear map. Suppose that 
$$\| \Xi \|_{(1, \alpha)\to (1, 1, \alpha)} = t_\alpha, \qquad \alpha=1, 2.$$
Then for every $\alpha\in [1, 2]$ we have
$$\| \Xi \|_{(1, \alpha)\to (1, 1, \alpha)} \leq t_1^{1-\theta} \cdot t_2^{\theta},$$
where $\theta= 2/\alpha'$.
\end{corollary}

%%%%%%%%%%%%%%%%%%%%%%%%%%%%%%%%%%%%%%%%%%%%%%%
\section{A new Tsallis  mutual information}\label{app:MI}
Given a convex function $\mathsf f$ satisfying $\mathsf f(1)=0$ and two distributions $p(x)$ and $q(x)$ on a discrete space $\mathcal{X}$, the $\mathsf  f$-divergence between $p$ and $q$ is defined as
$$D_{\mathsf f}(p\|q) = \sum_{x}q(x)\mathsf  f\left(\frac{p(x)}{q(x)}\right).$$
There are two proposals for defining a mutual information in terms of such a divergence. 
The first one given in~\cite[Eq. 3.10.1]{cohen1998comparisons} is
$$I^{CKZ}_{\mathsf f}(A;B):= D_{\mathsf f}(p_{AB}\|p_A\times p_B)=\sum_{a,b}p(a)p(b)\mathsf f\left(\frac{p(ab)}{p(a)p(b)}\right)=\sum_{a}p(a)D_{\mathsf f}(p_{B|a}\|p_B),$$
and has been studied in the literature (e.g. see \cite[Theorem 5.2]{raginsky2016strong},\cite{hsu2018generalizing}). Another definition is given in  \cite[Eq. 79]{polyanskiy2010arimoto}:
$$I^{PV}_{\mathsf f}(A;B):= \min_{q_B}D_{\mathsf f}(p_{AB}\|p_A\times q_B).$$

Herein, we propose yet a new definition of  mutual $\mathsf f$-information. Given a convex function $\mathsf f$, we define its mutual $\mathsf f$-information by
\begin{align}
I_{\mathsf f}(A;B)&:=\min_{q_B}D_{\mathsf f}(p_{AB}\|p_A\times q_B)- D_{\mathsf f}(p_B\|q_B)\nonumber
\\&=
\min_{q_B}\sum_{a}p(a)D_{\mathsf f}(p(_{B|a}\|q_B)- D_{\mathsf f}(p_B\|q_B).\label{eqnd434fkjht}
\end{align}
Observe that our mutual $\mathsf f$-information is smaller then the previous ones:
$$I_{\mathsf f}(A;B)\leq I_{\mathsf f}^{PV}(A;B)\leq I_{\mathsf f}^{CKZ}(A;B).$$
Moreover, when $D_{\mathsf f}(\cdot\|\cdot)$ is the KL divergence, $I_{\mathsf f}(A;B)$ reduces to Shannon's mutual information. 

Since we expect mutual $\mathsf f$-information to satisfy the data processing inequality, we impose a further assumption on the convex function $\mathsf f$. Interestingly, this assumption is the same as the one that gives the subadditivity of the $\Phi$-entropy.

\begin{mydef}\label{def:class-C}
Define $\mathscr F$ be the class of convex functions $\mathsf f(t)$ on $[0,\infty]$ that are \emph{not} affine (not of the form $t\mapsto at+b$ for some constants $a$ and $b$), $\mathsf f(1)=0$,  and $1/\mathsf f''$ is concave. 
\end{mydef}

An important property of the class $\mathscr F$ is the following.
%~\cite[Exercise 14.2]{boucheron2004concentration}.

\begin{lemma}\label{lemma1sd} For any function $\mathsf f\in\mathcal{F}$ and non-negative weights $\lambda_i$, $i=1,2,\cdots, n$, adding up to one, the function
$$G(x_1, x_2, \cdots, x_n)= \sum_i \lambda_i \mathsf f\left(x_i\right) -\mathsf f\Big(\sum_{i}\lambda_i x_i\Big) $$
is jointly convex. 
\end{lemma}
\begin{proof}
From  \cite[Exercise 14.2]{boucheron2004concentration}, concavity of $1/\mathsf f''$ implies that the  function $(s, t)\mapsto \lambda \mathsf f(s) + (1-\lambda) \mathsf f(t) - \mathsf f(\lambda s+ (1-\lambda) t) $, for any $\lambda\in[0,1]$, is jointly convex. One can use induction to obtain the claim of this lemma.
\end{proof}

Examples of functions in $\mathscr F$ include $\mathsf f(t)=t\log t$ and $\mathsf f(t)=\frac{1}{\alpha-1}(t^{\alpha}-1)$ for $\alpha\in(1, 2]$.

\begin{theorem}\label{data processing T} 
For any function $\mathsf f\in\mathscr F$, the mutual $f$-information $I_{\mathsf f}(A;B)$
satisfies the followings:
\begin{itemize}
\item[{\rm (i)}] $I_{\mathsf f}(A;B)=0$ if and only if $A$ and $B$ are independent. 
\item[{\rm (ii)}] If $C-A-B-D$ forms a Markov chain, then $I_{\mathsf f}(A;B)\geq I_{\mathsf f}(C;D)$.
\end{itemize}
\end{theorem}

\begin{proof}
The proof of (i) is easy, so only present the proof of (ii).
Observe that 
$$I_{\mathsf f}(A;B)=\min_{q_B}\sum_{b}q(b)\left(\sum_a p(a)\mathsf f\left(\frac{p(b|a)}{q(b)}\right) -\mathsf f\left(\frac{p(b)}{q(b)}\right)\right).$$
Take some Markov chain $C-A-B$. 
%We have
%$$I_{\mathsf f}(C;B)=\min_{q(b)}\left\{\sum_{b}q(b)\left[\sum_c p(c)\mathsf f\left(\frac{p(b|c)}{q(b)}\right) -\mathsf f\left(\frac{p(b)}{q(b)}\right)\right]\right\}.$$
Since $\mathsf f$ is convex, by Jensen's inequality we have
\begin{align*}\sum_a p(a)\mathsf f\left(\frac{p(b|a)}{q(b)}\right)&=\sum_{a,c} p(a,c)\mathsf f\left(\frac{p(b|a)}{q(b)}\right)
\\&\geq \sum_{c} p(c)\mathsf f\left(\sum_{a}p(a|c)\frac{p(b|a)}{q(b)}\right)
\\&=\sum_c p(c)\mathsf f\left(\frac{p(b|c)}{q(b)}\right).
\end{align*}
Therefore, $I_{\mathsf f}(C;B)\leq I_{\mathsf f}(A;B)$.

Next,  take some Markov chain $A-B-D$.  %We have
%$$I_{\mathsf f}(A;D)=\min_{q(d)}\left\{\sum_{d}q(d)\left[\sum_a p(a)\mathsf f\left(\frac{p(d|a)}{q(d)}\right) -\mathsf f\left(\frac{p(d)}{q(d)}\right)\right]\right\}.$$
To prove
$$I_{\mathsf f}(A;B)\geq I_{\mathsf f}(A;D),$$
it suffices to take some $q(b)$ and introduce some $q(d)$ such that
\begin{align}
\sum_{b}q(b)\left[\sum_a p(a)\mathsf f\left(\frac{p(b|a)}{q(b)}\right) -\mathsf f\left(\frac{p(b)}{q(b)}\right)\right]
\geq
\sum_{d}q(d)\left[\sum_a p(a)\mathsf f\left(\frac{p(d|a)}{q(d)}\right) -\mathsf f\left(\frac{p(d)}{q(d)}\right)\right]\label{eqn:A32}
\end{align}
Let $q(d)=\sum_{b}q(b)p(d|b)$. 
Then, we can write
$$\sum_{b}q(b)\left[\sum_a p(a)\mathsf f\left(\frac{p(b|a)}{q(b)}\right) -\mathsf f\left(\frac{p(b)}{q(b)}\right)\right]
=
\sum_{b,d}q(d)p(d|b)\left[\sum_a p(a)\mathsf f\left(\frac{p(b|a)}{q(b)}\right) -\mathsf f\left(\frac{p(b)}{q(b)}\right)\right]
$$
To prove \eqref{eqn:A32}, it suffices to show that for every $d$, we have
\begin{align}\sum_{b}p(d|b)\left[\sum_a p(a)\mathsf f\left(\frac{p(b|a)}{q(b)}\right) -\mathsf f\left(\frac{p(b)}{q(b)}\right)\right]
\geq
\sum_a p(a)\mathsf f\left(\frac{p(d|a)}{q(d)}\right) -\mathsf f\left(\frac{p(d)}{q(d)}\right)\label{eqn:A4352}\end{align}
For a fixed and given $p(a)$, consider the function
$$g\in \ell(\mathcal A)~ \mapsto~ \sum_a p(a)\mathsf{f} \big(g(a)\big) -\mathsf f\bigg(\sum_{a}p(a)g(a)\bigg).$$
According to Lemma \ref{lemma1sd}, this function is jointly convex in $g$. Therefore, \eqref{eqn:A4352} follows from Jensen's inequality on this jointly convex function since
$$\frac{p(d|a)}{q(d)}=\sum_b p(d|b) \frac{p(b|a)}{q(d)}.$$
\end{proof}

%%%

Let $\mathsf f_{\alpha}(t)=\frac{1}{\alpha-1}(t^{\alpha}-1)$ for $\alpha\in(1,2]$. As mentioned above this function belongs to $\mathscr F$. Then, following \eqref{eqnd434fkjht} we can define the Tsallis mutual information of order $\alpha$ by
\begin{align}\label{eq:tmi-0}
I_{\mathsf f_\alpha}(A;B)=\frac{1}{\alpha-1}\min_{q_B}\bigg\{\sum_{a}p(a)\bigg(\sum_{b}q(b)^{1-\alpha}p(b|a)^\alpha\bigg) -  \sum_{b}q(b)^{1-\alpha}p(b)^\alpha\bigg\}.
\end{align}
The reason that we call it Tsallis mutual information is that the Tsallis relative entropy can be defined in terms of the function $\mathsf f_{\alpha}$.

\begin{theorem} \label{Tsalis} 
The Tsallis mutual information defined in~\eqref{eq:tmi-0} equals 
\begin{align*}
I_{\mathsf f_\alpha}(A;B)=\frac{1}{\alpha-1}\bigg(\sum_b \bigg(\sum_a p(a) p(b|a)^\alpha - p(b)^\alpha\bigg)^\frac{1}{\alpha}\bigg)^\alpha.
\end{align*}
In particular, we have
\begin{align*}
\sqrt{I_{\mathsf f_2}(A;B)}&=\sum_b \bigg(\sum_a p(a) p(b|a)^2 - p(b)^2\bigg)^\frac{1}{2}
\\&=\sum_b \bigg(\sum_a p(a) \Big(p(b|a) - p(b)\Big)^2\bigg)^\frac{1}{2}
\\&=V_2(A;B).
\end{align*}
\end{theorem}

\begin{proof}
We use the Lagrange multipliers method for the optimal $q_B$ in~\eqref{eq:tmi-0}. The Lagrangian function of the optimization problem  equals
\begin{align*}
\mathcal{L}(q_B,\lambda)&= \frac{1}{\alpha-1}\bigg(\sum_{a,b}p(a)q(b)^{1-\alpha}p(b|a)^\alpha-  \sum_{b}q(b)^{1-\alpha}p(b)^\alpha\bigg)-\lambda \Big( \sum_b q(b) -1\Big)\\
&=\frac{1}{\alpha-1}\bigg(\sum_{a,b}p(a)q(b)^{1-\alpha}\Big( p(b|a)^\alpha- p(b)^\alpha\Big)\bigg)-\lambda \Big( \sum_b q(b) -1\Big)
\\&=\frac{1}{\alpha-1}\bigg(\sum_{b}q(b)^{1-\alpha}g(b)\bigg)-\lambda \bigg( \sum_b q(b) -1\bigg),
\end{align*}
where $g(b)=\sum_a p(a) p(b|a)^\alpha - p(b)^\alpha$. Note that by Jensen's inequality we have
$g(b)\geq 0$ for all $b\in \mathcal B$. The function $\mathcal{L}(q_B,\lambda)$ is convex in $q_B$. Then to find its minimum with respect to $q_B$, we take the derivative:
\begin{align*}
\frac{\partial\mathcal{L}(q_B,\lambda)}{\partial q(b)}=-q(b)^{-\alpha}g(b)-\lambda=0.
\end{align*}
This shows that $q(b)$ must be proportional to $g(b)^{\frac{1}{\alpha}}$. Then the optimal $q_B$ is given by
\begin{align*}
q^*(b)=\frac{g(b)^\frac{1}{\alpha}}{\sum_{\bar{b}}g(\bar{b})^\frac{1}{\alpha}}.
\end{align*}
Substituting $q(b)=q^*(b)$ in~\eqref{eq:tmi-0} yields the desired result.
\end{proof}

%******************************************
\section{Semantic security and $V_\infty$}\label{app:semantic-sec}

Most existing works in information theoretic security literature assume a message that is random and uniformly distributed. However, as pointed out in \cite{CuffPermuter} this assumption may not be valid for many real-life messages such as files or votes. The semantic security is a cryptographic requirement that addresses this point. It was shown in  \cite{bellare2012semantic} that semantic security is equivalent with a negligible mutual information between the message and the adversary's observations \emph{for all message distributions}. 

For a bipartite probability distribution $p_{AB}$ we have
$$V_\infty(A;B) = \sum_{b} \max_{a:\, p(a)>0} \big|p(b|a) - p(b) \big|.$$
This expression is similar to $V_1(A;B)$, except that the average over $a\in \mathcal A$ is replaced by a maximum over $a$. We show that $V_\infty(A;B)$ is related to the semantic security. If $A$ is the message and $B$ is an eavesdropper's information, $V_1(A;B)$ can be understood as the average leakage (over all messages), whereas $V_\infty(A;B)$ controls the worst-case leakage. That is, if $V_\infty(A;B)$ is small,  any two distinct message symbols $a, a'$ cannot
be distinguished by the eavesdropper. However, if $V_1(A;B)$ is small, it may be still the case that few
of the message symbols are perfectly distinguishable. 

Given $p_{AB}$, the authors in \cite{bellare2012semantic} show that semantic security holds if and only if $I(A;B)_q$ is small for all $q_{AB}$ of the form $q_{AB}=q_A\times p_{B|A}$ where $q_A$ is an arbitrary input distribution on $\mathcal{A}$. We claim that for any $q_{AB}=q_A\times p_{B|A}$ we have
\begin{align}I(A; B)_q\leq 2\log(e)V_\infty(A;B).
\label{eqn-sem-sec}
\end{align}
Therefore, if $V_{\infty}(A; B)$ is small, semantic security is guaranteed. 
This establishes the connection between our measure of correlation and semantic security. 
Note that
\begin{align}
I(A;B)_q&=\sum_{a,b}q(a,b)\log\frac{p(b|a)}{q(b)}\nonumber
\\&=\sum_{a,b}q(a,b)\log\Big(1+\frac{p(b|a)-q(b)}{q(b)}\Big)\nonumber
\\&\leq \sum_{a,b}q(a,b)\log(e) \frac{|p(b|a)-q(b)|}{q(b)}\label{logeq}
\\&=\log(e) \sum_{a,b}q(a|b) \big|p(b|a)-q(b)\big|\nonumber
\\&\leq \log(e)\sum_{a,b}q(a|b)\max_{{a'}}\big|p(b| a')-q(b)\big|\nonumber
\\&= \log(e)\sum_{b}\max_{a'}\big|p(b|a')-q(b)\big|,\nonumber
\end{align}
where in \eqref{logeq} we used the inequality $\log(1+x)\leq \log(e)|x|$ for $x>-1$. 
Using the triangle inequality we continue 
\begin{align*}
I(A;B)_q & \leq   \log(e)\sum_{b}\max_{a'}\Big|p(b|a')-\sum_a q(a)p(b|a)\Big|\nonumber\\
& \leq  \log(e)\sum_{b}\max_{a'}  \sum_a q(a)\Big|p(b|a')-p(b|a)\Big|\nonumber\\
& \leq   \log(e)\sum_{b}\max_{a, a'}  \big|p(b|a')-p(b|a)\big|\nonumber\\
& \leq  \log(e)\sum_{b}\max_{a, a'}  \big|p(b|a')-p(b)\big| + \big|p(b|a)-p(b)\big|\nonumber\\
& = 2 \log(e) V_\infty(A; B).
\end{align*}

%%%%%%%%%%%%%%%%%%%%%%
%%%% Bibliography %%%%

\bibliographystyle{IEEEtran}

\begin{thebibliography}{10}
\providecommand{\url}[1]{#1}
\csname url@samestyle\endcsname
\providecommand{\newblock}{\relax}
\providecommand{\bibinfo}[2]{#2}
\providecommand{\BIBentrySTDinterwordspacing}{\spaceskip=0pt\relax}
\providecommand{\BIBentryALTinterwordstretchfactor}{4}
\providecommand{\BIBentryALTinterwordspacing}{\spaceskip=\fontdimen2\font plus
\BIBentryALTinterwordstretchfactor\fontdimen3\font minus
  \fontdimen4\font\relax}
\providecommand{\BIBforeignlanguage}[2]{{%
\expandafter\ifx\csname l@#1\endcsname\relax
\typeout{** WARNING: IEEEtran.bst: No hyphenation pattern has been}%
\typeout{** loaded for the language `#1'. Using the pattern for}%
\typeout{** the default language instead.}%
\else
\language=\csname l@#1\endcsname
\fi
#2}}
\providecommand{\BIBdecl}{\relax}
\BIBdecl

\bibitem{DBWR14}
F.~Dupuis, M.~Berta, J.~Wullschleger, and R.~Renner, ``One-shot decoupling,''
  \emph{Communications in Mathematical Physics}, vol. 328, no.~1, pp. 251--284,
  2014.

\bibitem{yassaee2014achievability}
M.~H. Yassaee, M.~R. Aref, and A.~Gohari, ``Achievability proof via output
  statistics of random binning,'' \emph{IEEE Transactions on Information
  Theory}, vol.~60, no.~11, pp. 6760--6786, 2014.

\bibitem{liang2009information}
Y.~Liang, H.~V. Poor, S.~Shamai \emph{et~al.}, ``Information theoretic
  security,'' \emph{Foundations and Trends{\textregistered} in Communications
  and Information Theory}, vol.~5, no. 4--5, pp. 355--580, 2009.

\bibitem{yagli2018exact}
S.~Yagli and P.~Cuff, ``Exact soft-covering exponent,'' \emph{arXiv preprint
  arXiv:1801.00714}, 2018.

\bibitem{Issa}
I.~Issa and A.~B. Wagner, ``Measuring secrecy by the probability of a
  successful guess,'' \emph{IEEE Transactions on Information Theory}, vol.~63,
  no.~6, pp. 3783--3803, June 2017.

\bibitem{Kamath}
I.~Issa, S.~Kamath, and A.~B. Wagner, ``An operational measure of information
  leakage,'' in \emph{2016 Annual Conference on Information Science and Systems
  (CISS)}, March 2016, pp. 234--239.

\bibitem{Li}
C.~T. Li and A.~E. Gamal, ``Maximal correlation secrecy,'' in \emph{2015 IEEE
  International Symposium on Information Theory (ISIT)}, June 2015, pp.
  2939--2943.

\bibitem{Cuff}
C.~Schieler and P.~Cuff, ``Rate-distortion theory for secrecy systems,''
  \emph{IEEE Transactions on Information Theory}, vol.~60, no.~12, pp.
  7584--7605, Dec 2014.

\bibitem{Weinberger}
N.~Weinberger and N.~Merhav, ``A large deviations approach to secure lossy
  compression,'' \emph{IEEE Transactions on Information Theory}, vol.~63,
  no.~4, pp. 2533--2559, April 2017.

\bibitem{bellare2012semantic}
M.~Bellare, S.~Tessaro, and A.~Vardy, ``Semantic security for the wiretap
  channel,'' in \emph{Advances in Cryptology--CRYPTO 2012}.\hskip 1em plus
  0.5em minus 0.4em\relax Springer, 2012, pp. 294--311.

\bibitem{dodis2005entropic}
Y.~Dodis and A.~Smith, ``Entropic security and the encryption of high entropy
  messages,'' in \emph{Theory of Cryptography Conference}.\hskip 1em plus 0.5em
  minus 0.4em\relax Springer, 2005, pp. 556--577.

\bibitem{yu2017r}
L.~Yu and V.~Y. Tan, ``R{\'e}nyi resolvability and its applications to the
  wiretap channel,'' \emph{arXiv preprint arXiv:1707.00810}, 2017.

\bibitem{Bhatia-M}
R.~Bhatia, \emph{Matrix Analysis}, ser. Graduate Texts in Mathematics {\bf
  169}.\hskip 1em plus 0.5em minus 0.4em\relax Springer, 1997.

\bibitem{Pisier}
G.~Pisier, \emph{Non-commutative Vector Valued $L_p$-spaces and Completely
  $p$-summing Maps}.\hskip 1em plus 0.5em minus 0.4em\relax Soci\'et\'e
  Math\'ematique de France, 1998.

\bibitem{Junge96}
M.~Junge, ``Factorization theory for spaces of operators,'' \emph{Habilitation
  thesis Kiel University}, 1996.

\bibitem{DJKR}
I.~Devetak, M.~Junge, C.~King, and M.~B. Ruskai, ``Multiplicativity of
  completely bounded p-norms implies a new additivity result,''
  \emph{Communications in mathematical physics}, vol. 266, no.~1, pp. 37--63,
  2006.

\bibitem{verdu2015alpha}
S.~Verd{\'u}, ``$\alpha$-mutual information,'' in \emph{Information Theory and
  Applications Workshop (ITA), 2015}.\hskip 1em plus 0.5em minus 0.4em\relax
  IEEE, 2015, pp. 1--6.

\bibitem{BD}
P.~Delgosha and S.~Beigi, ``Impossibility of local state transformation via
  hypercontractivity,'' \emph{Communications in Mathematical Physics}, vol.
  332, no.~1, pp. 449--476, 2014.

\bibitem{Beigi}
S.~Beigi, ``Sandwiched r\'enyi divergence satisfies data processing
  inequality,'' \emph{Journal of Mathematical Physics}, vol.~54, no.~12, p.
  122202, 2013.

\bibitem{sason2016f}
I.~Sason and S.~Verd{\'u}, ``$f$-divergence inequalities,'' \emph{IEEE
  Transactions on Information Theory}, vol.~62, no.~11, pp. 5973--6006, 2016.

\bibitem{sason2015upper}
------, ``Upper bounds on the relative entropy and r{\'e}nyi divergence as a
  function of total variation distance for finite alphabets,'' in
  \emph{Information Theory Workshop-Fall (ITW), 2015 IEEE}.\hskip 1em plus
  0.5em minus 0.4em\relax IEEE, 2015, pp. 214--218.

\bibitem{Sharma15}
N.~Sharma, ``Random coding exponents galore via decoupling,''
  \emph{arXiv[quant-ph]: 1504.07075}, 2015.

\bibitem{Berta08}
M.~Berta, ``Single-shot quantum state merging,'' Diploma thesis, ETH Z\"urich,
  arXiv[quant-ph]:0912.4495, 2008.

\bibitem{Renner}
\BIBentryALTinterwordspacing
R.~Renner, ``Security of quantum key distribution,'' \emph{International
  Journal of Quantum Information}, vol.~06, no.~01, pp. 1--127, 2008. [Online].
  Available:
  \url{http://www.worldscientific.com/doi/abs/10.1142/S0219749908003256}
\BIBentrySTDinterwordspacing

\bibitem{RK05}
R.~Renner and R.~K\"onig, ``Universally composable privacy amplification
  against quantum adversaries,'' in \emph{Second Theory of Cryptography
  Conference TCC, volume 3378 of Lecture Notes in Computer Science}.\hskip 1em
  plus 0.5em minus 0.4em\relax springer, 2005.

\bibitem{hayashi2011exponential}
M.~Hayashi, ``Exponential decreasing rate of leaked information in universal
  random privacy amplification,'' \emph{IEEE Transactions on Information
  Theory}, vol.~57, no.~6, pp. 3989--4001, 2011.

\bibitem{chung1986some}
F.~R. Chung, R.~L. Graham, P.~Frankl, and J.~B. Shearer, ``Some intersection
  theorems for ordered sets and graphs,'' \emph{Journal of Combinatorial
  Theory, Series A}, vol.~43, no.~1, pp. 23--37, 1986.

\bibitem{BDL16}
S.~Beigi, N.~Datta, and F.~Leditzky, ``Decoding quantum information via the
  petz recovery map,'' \emph{Journal of Mathematical Physics}, vol.~57, p.
  082203, 2016.

\bibitem{HHWY08}
P.~Hayden, M.~Horodecki, A.~Winter, and J.~Yard, ``A decoupling approach to the
  quantum capacity,'' \emph{Open Systems \& Information Dynamics}, vol.~15,
  no.~01, pp. 7--19, 2008.

\bibitem{csiszar1995generalized}
I.~Csisz{\'a}r, ``Generalized cutoff rates and r{\'e}nyi's information
  measures,'' \emph{IEEE Transactions on information theory}, vol.~41, no.~1,
  pp. 26--34, 1995.

\bibitem{parizi2017exact}
M.~B. Parizi, E.~Telatar, and N.~Merhav, ``Exact random coding secrecy
  exponents for the wiretap channel,'' \emph{IEEE Transactions on Information
  Theory}, vol.~63, no.~1, pp. 509--531, 2017.

\bibitem{hayashi2006general}
M.~Hayashi, ``General nonasymptotic and asymptotic formulas in channel
  resolvability and identification capacity and their application to the
  wiretap channel,'' \emph{IEEE Transactions on Information Theory}, vol.~52,
  no.~4, pp. 1562--1575, 2006.

\bibitem{endo2014reliability}
T.~S. Han, H.~Endo, and M.~Sasaki, ``Reliability and secrecy functions of the
  wiretap channel under cost constraint,'' \emph{IEEE Transactions on
  Information Theory}, vol.~60, no.~11, pp. 6819--6843, 2014.

\bibitem{6034266}
O.~Shayevitz, ``On r{\'e}nyi measures and hypothesis testing,'' in \emph{2011
  IEEE International Symposium on Information Theory Proceedings}, July 2011,
  pp. 894--898.

\bibitem{Lunardi}
A.~Lunardi, \emph{Interpolation Theory}, ser. Lecture Notes (Scuola Normale
  Superiore).\hskip 1em plus 0.5em minus 0.4em\relax Scuola Normale Superiore,
  2009, vol.~9.

\bibitem{cohen1998comparisons}
J.~Cohen, J.~Kempermann, and G.~Zbaganu, \emph{Comparisons of Stochastic
  Matrices with Applications in Information Theory, Statistics, Economics and
  Population}.\hskip 1em plus 0.5em minus 0.4em\relax Springer Science \&
  Business Media, 1998.

\bibitem{raginsky2016strong}
M.~Raginsky, ``Strong data processing inequalities and $\phi $-sobolev
  inequalities for discrete channels,'' \emph{IEEE Transactions on Information
  Theory}, vol.~62, no.~6, pp. 3355--3389, 2016.

\bibitem{hsu2018generalizing}
H.~Hsu, S.~Asoodeh, S.~Salamatian, and F.~P. Calmon, ``Generalizing bottleneck
  problems,'' \emph{arXiv preprint arXiv:1802.05861}, 2018.

\bibitem{polyanskiy2010arimoto}
Y.~Polyanskiy and S.~Verd{\'u}, ``Arimoto channel coding converse and r{\'e}nyi
  divergence,'' in \emph{Communication, Control, and Computing (Allerton), 2010
  48th Annual Allerton Conference on}.\hskip 1em plus 0.5em minus 0.4em\relax
  IEEE, 2010, pp. 1327--1333.

\bibitem{boucheron2004concentration}
S.~Boucheron, G.~Lugosi, and O.~Bousquet, \emph{Concentration
  inequalities}.\hskip 1em plus 0.5em minus 0.4em\relax Springer, 2004.

\bibitem{CuffPermuter}
Z.~Goldfeld, P.~Cuff, and H.~H. Permuter, ``Semantic-security capacity for
  wiretap channels of type ii,'' \emph{IEEE Transactions on Information
  Theory}, vol.~62, no.~7, pp. 3863--3879, July 2016.

\end{thebibliography}

\end{document}